\documentclass[12pt,reqno]{article}

\usepackage[hmargin=2cm,bmargin=2cm,tmargin=2.5cm]{geometry}

\usepackage[T1]{fontenc}
\usepackage{enumitem,kantlipsum}
\usepackage[leqno]{amsmath}%
\usepackage{amsthm}
\usepackage{amsxtra}%
\usepackage{amsfonts}%
\usepackage{amssymb}%
\usepackage{color,hyperref}
\usepackage{graphicx}
\usepackage{tikz}
\usepackage{subfig}
\usepackage{cite}
\hypersetup{colorlinks,breaklinks,
	linkcolor=blue,urlcolor=blue,
	anchorcolor=blue,citecolor=blue}
\usepackage{color}
\usepackage{array}
\usepackage{booktabs}
\usepackage{dsfont}
\usepackage[toc,page]{appendix}
\usepackage{authblk}
\usepackage{booktabs}
\usepackage[font=small]{caption}
\usepackage{microtype}
\usepackage{nicefrac}
\usepackage[T1]{fontenc}
\usepackage{lmodern}
\usepackage{authblk}

\usepackage{graphicx}
\usepackage{color}
\usepackage{enumerate}
\usepackage{enumitem,kantlipsum}
\usepackage[english]{babel}

\usepackage{amsmath,amssymb}
\usepackage{bm}
\usepackage{amsthm}
\usepackage{mathtools}
\usepackage[utf8]{inputenc}
\usepackage[T1]{fontenc}
\usepackage[b]{esvect}
\usepackage{amssymb}
\usepackage{algorithmicx}
\usepackage[ruled]{algorithm}
\usepackage{algpseudocode}
\usepackage{parskip}
\usepackage{enumerate}
\usepackage{bm}
\usepackage{array}
\usepackage{arydshln}
\usepackage{multirow}
\usepackage{bigstrut}
\usepackage{bbm}
\usepackage{esdiff}
\usepackage{cleveref}
\usepackage{caption}
\usepackage{lettrine}
\usepackage{yfonts,color}

\usepackage{relsize}
\usepackage{environ}
\NewEnviron{myequation}{%
	\begin{equation*}
	\scalebox{.775}{$\BODY$}
	\end{equation*}
}
\NewEnviron{myequation2}{%
	\begin{equation*}
	\scalebox{.85}{$\BODY$}
	\end{equation*}
}
\newtheorem{thm}{Theorem}[section]
\newtheorem{lem}[thm]{Lemma}
\newtheorem{prop}[thm]{Proposition}
\newtheorem{cor}[thm]{Corollary}

\newtheorem{assum}{Assumption}

\theoremstyle{definition}
\newtheorem{defn}{Definition}
\newtheorem{exmp}{Example}
\newtheorem{remark}[thm]{Remark}

\newcommand{\cc}[1]{\overline{#1}}

\newcommand{\cD}{\mathcal{D}}

\newcommand{\cH}{\mathcal{H}}

\newcommand{\cM}{\mathcal{M}}
\newcommand{\cN}{\mathcal{N}}
\newcommand{\cO}{\mathcal{O}}
\newcommand{\cP}{\mathcal{P}}
\newcommand{\cQ}{\mathcal{Q}}
\newcommand{\cR}{\mathcal{R}}
\newcommand{\cS}{\mathcal{S}}

\newcommand{\cU}{\mathcal{U}}

\newcommand{\cW}{\mathcal{W}}

\newcommand{\bC}{\mathbb{C}}

\newcommand{\bI}{\mathbb{I}}

\newcommand{\bN}{\mathbb{N}}

\newcommand{\bR}{\mathbb{R}}

\newcommand{\bT}{\mathbb{T}}

\newcommand{\bZ}{\mathbb{Z}}

\newcommand{\tr}{\mathrm{tr}}

\renewcommand{\Re}{\operatorname{Re}}

\newcommand{\vt}{\mathfrak{v}}

\newcommand{\Ltwo}{{L^2(\Gamma)}}

\newcommand{\eps}{\varepsilon}

\def\up#1{^{(#1)}}

\title{On the Spectra of Periodic Elastic Beam Lattices: \\Single-Layer Graph}

\author{Mahmood Ettehad\thanks{corresponding author}}
\affil{Institute for Mathematics and its Applications (IMA)\\ University of Minnesota}
\author{Burak Hat\.{i}no\u{g}lu}
\affil{Department of Mathematics \\ University of California, Santa Cruz}

\date{}
\begin{document}
	\maketitle
	
	\begin{abstract}
		We present full description of spectra for a Hamiltonian defined on periodic hexagonal elastic lattices. These continua are constructed out of Euler-Bernoulli beams, each governed by a scalar-valued self-adjoint operator, which is also known as the fourth order Schr\"{o}dinger operator, equipped with a real periodic symmetric potential. In contrast to the second order Schr\"{o}dinger operator commonly applied in quantum graph literature, here vertex matching conditions encode geometry of the underlying graph by their dependence on angles at which edges are met. We show that for a special equal angle lattice, known as graphene, dispersion relation has a similar structure as reported for the periodic second order Schr\"{o}dinger operator on hexagonal lattices. This property is then further utilized to prove existence of singular Dirac points. We further discuss reducibility of Fermi surface at uncountably many low-energy levels for this special lattice. Applying perturbation analysis, we extend the developed theory to derive dispersion relation for angle-perturbed Hamiltonian of hexagonal lattices in a geometric neighborhood of graphene. In these graphs, unlike graphene, dispersion relation is not splitted into purely energy and quasimomentum dependent terms, however singular Dirac points exist similar to the graphene case. 
	\end{abstract}
	
	\section{Introduction}
	Lattice materials are cellular structures obtained by tessellating a unit cell comprising a few beams. Such lattice materials exhibit the characteristic of pass and stop bands determining frequency intervals over which wave motion can or can not occur, respectively \cite{RSS03,KL15, LB19}. This unique directional behavior complements the stop-pass band pattern and makes the application of 2D periodic structures as directional mechanical filters \cite{RSS03}. For models on special lattices, e.g. graphene, interesting physical and spectral properties have been observed such as the presence of special conical points in the dispersion relation, where its different sheets touch to form a two-sided conical singularity \cite{SWF06, KP07, ZEKBSA17,BHJ19,BHJZ21}.
	
	The analysis of wave motion in periodic systems such as lattice materials and vibrations in harmonic atomic lattices are traced back to early studies of string vibration and later by Brillouin \cite{B53}. Under certain simplification assumptions, modeling variety of natural and engineered tessellated lattices can generally be studied under beam theories \footnote{most inclusive classical beam models are the Euler-Bernouli and Timoshenko beam theories.}. Under Euler-Bernouli beam model, each beam is described by an energy functional which involves four degrees of freedom for every infinitesimal element along the beam: axial, lateral (2 degrees of freedom) and angular displacements. At a joint, these four functions, supported on the beams involved, must be related via matching conditions that take into account the physics of a joint, see  \cite{BL04, GLL17, BE21} for more details. In the special case of the planar frames, the operator decomposes into a direct sum of two operators, one coupling out-of-plane to angular (torsional) displacements and the other coupling in-plane with axial displacements \cite{BE21}.
	
	From more theoretical point of view, recently the analysis of Hamiltonians corresponding to these symplectic structures has gained interest by mathematicians working on differential operators on metric graphs, see e.g. \cite{KKU15, GM20, BE21} and references therein. Along this line, early studies on derivation of dispersion relation (or variety) of second order {S}chr{\"o}dinger operator defined on a periodic graph, splits Hamiltonian into two essentially unrelated parts: the analysis on a single edge, and the spectral analysis on the combinatorial graph, the former being independent of the graph structure, and the latter independent of the potential \cite{KP07}. However, contrary to {S}chr{\"o}dinger type operator on graph, vertex conditions for beam Hamiltonian encode geometry by its dependence on the angles at which the edges are met. As a result, extension of the existing theory to the latter operator on periodic lattices is not trivially accessible.  
	
	The main focus of the current work is the extension of the reported results in \cite{KP07} to the fourth order operator $\cH = d^4/dx^4 + q(x)$ with self-adjoint vertex conditions and a real periodic symmetric potential on graphene and lattices in geometric neighborhood of it, see Figure \ref{fig:fundDomain}. This is done by considering the analysis of the operator $\cH$ on a single edge, and then the spectral analysis of $\cH$ on the combinatorial graph. The spectrum of the self-adjoint operator $\cH^{\text{per}} = d^4/dx^4 + q_0(x)$ on the real line with a real periodic potential (known as Hill operator for the second-order operator) has a band-gap structure and bounded below. In contrast to the Hill operator, the edges of the spectral bands may belong to not only the periodic or anti-periodic spectra of $\cH$ on $(0,1)$, but also the set of resonances \cite{BK10}. However the latter case may happen at most for finitely many bands \cite{BK10}. The resonances are the branch points of the Lyapunov function, which is an analytic function on a two-sheeted Riemann surface and depends on the monodromy matrix of $\cH$. The Lyapunov function characterizes the spectrum $\sigma(\cH)$ and multiplicities of its points. We refer interested reader to Section 2.2 of this work and \cite{BK05,BK10,BK12} for detailed discussions. 
	
	Before stating the structure of this paper, we briefly summarize the main results. In Theorem \ref{detProp2}, we obtain the dispersion relation of $\cH$ on graphene where it is shown that the absolutely continuous spectrum coincides with $\sigma(\cH)$ as a set and the singular continuous spectrum is empty. However the pure point spectrum is non-empty and coincides with the set of eigenvalues of $\cH$ on $(0,1)$ with Dirichlet boundary conditions and zero second derivative boundary conditions on both endpoints. Theorem \ref{grapheneSpectrum} describes these spectral properties of $\cH$ on graphene. In Theorem \ref{DiracPointsThm}, we prove a representation of the set of Dirac points (conical singularities) of the dispersion relation in terms of the two branches of the Lyapunov function. Then in Theorem \ref{ReducibleFermisurface}, we characterize reducible and irreducible Fermi surfaces. Investigation on the role of angle-dependent vertex conditions is done through perturbation analysis, where we present existence and stability of Dirac points under perturbed angles.  
	\begin{figure}[ht]
		\centering
		\includegraphics[width=0.325\textwidth]{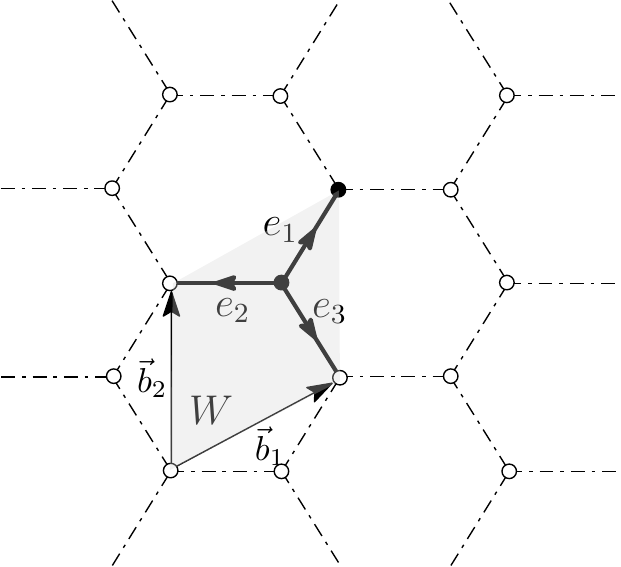}
		\caption{The hexagonal lattice $G$ and a fundamental domain $W$
			together with its set of vertices $V(W) = \{\vt_1,\vt_2\}$ and set of edges
			$E(W) = \{e_1,e_2,e_3\}$.}
		\label{fig:fundDomain}
	\end{figure}
	
	The paper is structured as follows: in Section \ref{sec:Preliminaries}, we summarize preliminary background starting with discussion on the parametrization of the beam deformation, energy functional, quadratic form and Hamiltonian on planar frames. This discussion is continued with a brief review of spectral properties of the fourth order periodic operator $\cH^{per}$ on the real line. In Section \ref{sec:SHLH}, we give a characterization of hexagonal elastic lattice’s Hamiltonian on graphene and its perturbations outside the Dirichlet spectrum. Section \ref{sec:GH} is devoted to the derivation of the dispersion relation, Dirac points, and spectral structure for the graphene lattice. Extension of the results for perturbed angles is the topic of Section \ref{sec:APH}. Section \ref{sec:Outlook} contains additional remarks and potential future extensions. 
	\section{Preliminaries}
	\label{sec:Preliminaries}
	In this section we will briefly review existing results in the literature to build necessary background for understanding the forthcoming materials. More specifically, in the first part self-adjoint beam operator $\cH$ on graph $G$ along with corresponding vertex conditions is defined. Next, we briefly discuss spectral results for similar type of periodic operator but defined on the real line,  $\cH^{\text{per}}$, known as Hill operator in the second-order case. We summarize these result from \cite{BL04, BK05, BK10} in Theorem \ref{summaryRefResults}, which will be repeatedly referred in the forthcoming sections. Reader familiar with these materials can safely skip this preliminary discussions and start with results in Section \ref{sec:SHLH}.
	\subsection{Elastic Planar Graphs}
	Under Euler-Bernouli beam model, each beam is described by an energy functional which involves four degrees of freedom for every infinitesimal element along the beam: axial, lateral (2 degrees of freedom) and angular displacements.  A central importance here is how to derive vertex matching conditions which are at the same time mathematically general and physically sound for application purposes. By restricting to one-degree of freedom, namely lateral displacement, vertex conditions for planar graphs have been derived by assuming that the deformed lattice will remain locally planar at vertex, i.e. existence of the tangent plane at that vertex \cite{KKU15}. It is shown that the resulting scalar-valued operator is self-adjoint. Extension of these results to generally three-dimensional graphs is developed in \cite{BE21}. This has been done by introducing the notion of rigidity at the vertex, on which matching conditions can be derived out of the first principle. Interestingly, the remaining vertex conditions which make vector-valued operator self-adjoint, have connection to the engineering world, namely satisfying equilibrium of force and moments at vertex. Further extensions of these result to semi-rigid type joint have been recently proposed in \cite{SBE21} where discontinuity of the displacement and rotation fields are admissible at a vertex. In a special case of planar frames, the operator decomposes into a direct sum of two operators, one coupling out-of-plane to angular (torsional) displacements and the other coupling in-plane with axial displacements. However achieving this level of physically sound models means that the operator is no longer scalar valued and contains at-least two degrees of freedom (for planar graphs) coupled at the joints. In this work we follow results of the scalar valued operator from \cite{KKU15} with the benefit of revealing some solid theoretical results regarding spectra of the corresponding Hamiltonian on periodic hexagonal lattice.  
	\subsubsection{\textbf{Energy Functional on Planar Lattice}}
	\label{sec:frames_description}
	A beam frame is a collection of beams connected at joints.
	We describe a beam frame as a geometric graph $\Gamma = (V,E)$,
	where $V$ denotes the set of vertices and $E$ the set of edges.  The
	vertices $\vt \in V$ correspond to joints and edges $e \in E$ are the
	beams.  Each edge $e$ is a collection of the following information:
	origin and terminus vertices $\vt_{e}^{o},\vt_{e}^{t}\in V$, length
	$\ell_e$ and the local basis
	$\{\vec{i}_e, \vec{j}_e, \vec{k}_e\}$. For special planar graphs $\vec{k}_e = \vec k$ for all edges $e \in E$, and thereby the graph $\Gamma$ can be embedded in $\bR^2$. Describing the
	vertices $V$ as points in $\bR^2$ also fixes the length $\ell_e$ and
	the axial direction $\vec{i}_e$ (from origin to terminus). However the choice
	of $\vec{j}_e$ in the plane orthogonal to $\vec{i}_e$ still needs to
	be specified externally. The distinction between origin and terminus, and thus the direction of
	$\vec i_e$ is unimportant in analysis but should be fixed for
	consistency.  It is important to use the same beam basis when writing
	out joint conditions at both ends of the beam.  We will use the
	\textit{incidence indicator} $s_{\vt}^e$ which is defined to be $1$
	when $\vt$ is the origin of $e$, $-1$ if it is the terminus of $e$ and
	$0$ otherwise. In the context of the kinematic Bernoulli assumptions for beam frame without pre-stress and external force, the total strain energy of the 
	beam frame is expressed as
	\begin{equation}
	\label{eq:energyFunc}
	\cU:= \frac{1}{2}  \sum_{e \in E}  \int_e \big(
	a_e |u''_e(x)|^2 +  q(x) |u_e(x)|^2  \big) dx.
	\end{equation}
	Above, parameter $a_e$ is positive and fixed over edge $e$ representing bending stiffness about the local axis $\vec j_e$ and $q \in L^2(e)$ is real-valued function.
	\begin{assum}
		The potential term $q \in L^2(e)$ is satisfying the evenness (symmetry) property 
		\begin{equation}
		\label{eq:evenPotential}
		q(x) = q(1-x).
		\end{equation}
		The evenness assumption \eqref{eq:evenPotential} is made not just for mathematical convenience, this condition is required if one considers operators invariant with respect to all symmetries of the periodic lattice.
	\end{assum}
	
	\subsubsection{\textbf{Quadratic and Operator Forms}}
	\label{sec:main_results}
	We now give a formal mathematical description of the Euler--Bernoulli
	strain energy form.
	\begin{thm}\emph{\textbf{(sesqulinear form \cite{KKU15,BE21})}}
		\label{primaryCond}
		Energy functional~\eqref{eq:energyFunc} of the planar beam lattice with
		free rigid joints is the quadratic form corresponding to the
		positive closed sesquilinear form
		\begin{align}
		\label{varEnergyForm}
		\cQ\left[u,\widetilde u\right] :=
		\sum_{e \in E} \int_e \big( a_e u_e''(x)\cc{\widetilde u_e''(x)} + q(x) u_e(x) \cc{\widetilde u_e(x)} \big)dx 
		\end{align}
		densely defined on the Hilbert space $\Ltwo := \mathop{\mathsmaller\bigoplus}_{e \in E} L^2(e)$ with the domain of $\cQ$ consisting of the vectors $\mathop{\mathsmaller\bigoplus}_{e \in E} H^2(e)$ that satisfy at every vertex $\vt \in V$  \textit{rigid joint}
		conditions, namely for all $e \sim \vt$
		\begin{subequations}
			\begin{gather}
			\label{dispRigid11}
			u_1(\vt) = \cdots = u_{n_\vt}(\vt) \\
			\label{rotationRigid11}
			(\vec j_2 \cdot \vec i_e) u_1'(\vt)
			+  (\vec j_e \cdot \vec i_1) u_2'(\vt)
			+  (\vec j_1 \cdot \vec i_2) u_e'(\vt)
			= 0
			\end{gather}
		\end{subequations}
	\end{thm}
	Above, all functions are evaluated at the vertex $\vt$ and all derivatives are taken in direction $\vec i_e$. We remark here that condition \eqref{rotationRigid11} guarantees to preserve the local structure of a planar frame at each vertex, see e.g. \cite{BL04, KKU15}. The following theorem characterizes the Hamiltonian of the frame as a self-adjoint differential operator on the metric graph.  
	\begin{thm}\emph{\textbf{(operator form \cite{KKU15})}}
		\label{MainTheorem}
		Energy form~\eqref{varEnergyForm} on a beam frame with free rigid
		joints corresponds to the self-adjoint operator
		$\cH \colon \Ltwo \to \Ltwo$ acting as
		\begin{equation}
		\label{eq:diffSystem}
		v_e
		\mapsto
		a_e u_e'''' + q u_e
		\end{equation}
		on every edge $e\in E$ of the graph. The domain of the operator $\cH$
		consists of the functions from $\mathop{\mathsmaller\bigoplus}_{e \in E} H^4(e)$
		that satisfy at each vertex $\vt \in V$:
		
		(i) primary conditions
		\begin{subequations}
			\begin{gather}
			\label{dispRigid1}
			u_1(\vt) = \cdots = u_{n_\vt}(\vt) \\
			\label{rotationRigid1}
			(\vec j_2 \cdot \vec i_e) u_1'(\vt)
			+  (\vec j_e \cdot \vec i_1) u_2'(\vt)
			+  (\vec j_1 \cdot \vec i_2) u_e'(\vt)
			= 0
			\end{gather}
		\end{subequations}
		
		(ii) conjugate conditions, namely for $e_\ell, e_{\ell'} \sim \vt$ such that $\vec i_\ell \times \vec i_{\ell'} \not = 0$
		\begin{subequations}
			\begin{gather}
			\label{secondBcThmDisp}
			\sum_{e \sim \vt} s_\vt^e a_e u_e'''(\vt)= \vec0 \\
			\label{eq:secondBcThmRota}
			\sum_{e \sim \vt} s_\vt^e a_e (\vec i_\ell \cdot \vec j_e) u_e'' (\vt) = 0 \qquad \& \qquad	\sum_{e \sim \vt} s_\vt^e a_e (\vec i_{\ell'} \cdot \vec j_e) u_e''(\vt) = 0
			\end{gather}
		\end{subequations}
	\end{thm}
	Thus the defined operator $\cH$ is unbounded and self-adjoint
	in the Hilbert space $L^2(\Gamma)$. Due to the condition \eqref{eq:evenPotential} on the potential, the Hamiltonian $\cH$ is invariant with respect to all symmetries
	of the hexagonal lattice  $\Gamma$, in particular with respect to the $\bZ^2$-shifts, which will play a crucial role in our considerations, see \cite{KP07} for a detailed discussion on the role of symmetry of the potential. 
	\subsection{Periodic fourth order Operator on the Real Line}\label{4thorderHill}
	Next we will summarize existing results on the spectrum of fourth-order operator with periodic potential on the real line. There are key differences compare to the second-order (Hill's) operator which are essential for us to develop our results. The reader familiar with the aforementioned discussions can skip this subsection and directly jump to Theorem \ref{summaryRefResults}. Consider the self-adjoint operator $\cH^{\text{per}} := d^4/dx^4 + q_0(x)$, acting on $L^2(\bR)$, where the real 1-periodic potential $q_0(x)$ belongs to the real space
	\begin{equation*}
	L_0^2(\bT) := \Big\{q_0 \in L^2(\bT) : \int_{0}^{1} q_0(x) dx = 0\Big\},
	\end{equation*}
	where $\bT = \bR / \bZ$. Introduce the fundamental solutions $\{g_k(x)\}_{k=1}^4$ of the eigenvalue problem 
	\begin{equation}
	\label{eq:rigrnValRealPer}
	\cH^{\text{per}} u(x) = \lambda u(x), \qquad (x,\lambda) \in \bR \times \bC
	\end{equation}
	satisfying for $j, k \in \{1,\ldots,4\}$ the conditions
	\begin{equation}\label{fundsolconditions}
	g_k^{(j-1)}(0)	 = \delta_{j k},
	\end{equation}
	where $\delta_{jk}$ is the Kronecker delta function and $g^{(k)} = d^kg/dx^k$. The monodromy matrix has the form $M(\lambda) := \cM(1,\lambda)$ in which for $x \in \bR$
	\begin{equation}
	\label{eq:monodMatrix}
	\cM(x,\lambda) := \{\cM_{j,k}(x,\lambda)\}_{j,k=1}^4 = \{g_k^{(j-1)}(x)\}_{j,k=1}^4
	\end{equation}
	and it shifts by the period along the solutions of \eqref{eq:rigrnValRealPer}. It is well-known that the monodromy matrix $M(\lambda)$ is entire on $\lambda$ and its eigenvalue $\tau \in \bC$, i.e. root of algebraic polynomial $D(\tau, \lambda) := \det(M(\lambda) - \tau \bI_4)$, is called a multiplier. According to the Lyapunov Theorem if for some $\lambda \in \mathbb{C}$, $\tau(\lambda)$ is a multiplier, then $\tau^{-1}(\lambda)$ is a multiplier of same multiplicity. Moreover, each $M(\lambda), \lambda \in \mathbb{C}$ has exactly four multipliers $\tau^{\pm1}_{1}(\lambda), \tau^{\pm 1}_2(\lambda)$\cite{BK05}. If we let $D_{\pm}(\lambda) = \frac{1}{4}D(\pm1,\lambda)$, then zeros of $D_+(\lambda)$ and $D_-(\lambda)$ are the eigenvalues of the periodic and anti-periodic problem respectively for \eqref{eq:rigrnValRealPer}. Denote by $\lambda_0^+, \lambda_{2n}^{\pm}$ and $\lambda_{2n-1}^\pm$ with $n \in \mathbb{N}$, the sequence of zeros of $D_+$ and $D_-$  (counted with multiplicity) respectively such that $\lambda_0^+ \leq \lambda_2^- \leq \lambda_2^+ \leq \lambda_4^- \leq \lambda_4^+ \leq \cdots$ and $\lambda_1^- \leq \lambda_1^+ \leq \lambda_3^- \leq \lambda_3^+ \leq \lambda_5^- \leq \cdots$. It is well known that the spectrum of $\cH^{\text{per}}$ is purely absolutely continuous and consists of non-degenerate intervals\cite{BK05,BK10}. These intervals are separated by the gaps $G_n = (E_n^-,E_n^+)$, $n \in \mathbb{N}$, with positive length. We introduce the functions 
	\begin{equation}
	T_1(\lambda) := \frac{1}{4} \tr\big(M(\lambda)\big), \qquad T_2 := \frac{1}{2} \big(\tr\big(M^2(\lambda)\big)+1\big) - \tr^2\big(M(\lambda)\big).
	\end{equation}
	
	The functions $T_1(\lambda) $, $T_2(\lambda)$ are entire, real on $\bR$ and
	\begin{equation}
	D(\tau,\cdot) = \big(\tau^2 - 2(T_1-T_2^{1/2})\tau+1\big)\big(\tau^2 - 2(T_1+T_2^{1/2})\tau+1\big).
	\end{equation}
	For the special case of the zero potential, i.e. $q_0 \equiv 0$, the corresponding functions have the form 
	\begin{equation}
	T_1^0(\lambda) = \frac{1}{2}\big(\cosh( \lambda^{1/4}) + \cos( \lambda^{1/4})\big), \qquad 
	T_2^0(\lambda) = \frac{1}{4}\big(\cosh( \lambda^{1/4}) - \cos( \lambda^{1/4})\big)^2,
	\end{equation}
	with $\arg (\lambda^{1/4}) \in (-\frac{\pi}{4}, \frac{\pi}{4}]$.
	Let $\{r_0^-, r_n^\pm\}_{n \in \bN}$ be the sequence of zeros of $T_2(\lambda)$ in $\bC$ (counted with multiplicity) such that $r_0^-$ is the maximal real zero, and $\cdots \leq \Re r_{n+1}^+ \leq \Re r_n^+ \leq \cdots \leq \Re r_1^+$. 
	If $r_n^{+} \in \mathbb{C}_{+}$, then $r_n^{-} = \overline{r_n^{+}} \in \mathbb{C}_{+}$, and if $r_n^{+} \in \mathbb{R}$, then $r_n^{-} \leq r_n^{+} \leq \Re r_m^{-}$ for $m = 1,\ldots, n$.
	Under extra mild conditions, then it has been shown that $r_n^\pm = -4(n \pi)^4 + \cO(n^2)$ as $n \rightarrow \infty$. Let $\cdots \leq r_{n_j}^- \leq r_{n_j}^+ \leq \cdots \leq r_{n_1}^- \leq r_{n_1}^+ \leq r_0^-$ be the subsequence of the real zeros of $T_2(\lambda)$, then $T_2(\lambda) < 0$ for any $\lambda \in R_j^0 := (r_{n_{j+1}}^+, r_{n_j}^-)$ for $j \in \bN$.
	\begin{defn}
		Following characterizations are in order
		\begin{itemize}
			\item A zero of the function $T_2(\lambda)$ is called a resonance of operator $\cH^{\text{per}}$.
			\item The interval $R_j^0 \subset \bR$ is called a resonance gap. 
		\end{itemize}
	\end{defn}
	Denote by $R^0 := \cup R_j^0$ and $\eta^0$ which joins the points $r_n^+, \overline{r_n^+}$ and does not cross $R^0$. To deal with the roots of the function $T_2(\lambda)$, the Riemann surface $\cR$ is constructed by taking two replicas of the $\lambda$-plane cut along $R^0$ and $\cup \eta_n$ and they are called sheets $\cR_1$ and $\cR_2$ respectively. As a result, there exists a unique analytic continuation of the function $T_2^{1/2}(\lambda)$. Let introduce Lyapunov function by 
	\begin{equation}
	\Delta(\xi) = T_1(\xi) + T_2^{1/2}(\xi)
	\end{equation}
	with $ \xi \in \cR$. Let $\Delta(\xi) = \Delta_1(\lambda)$ on $\cR_1$, and $\Delta(\xi) = \Delta_2(\lambda)$ on the second sheet $\cR_2$. Then 
	\begin{subequations}
		\label{eq:Delta12}
		\begin{gather}
		\Delta_1(\lambda) = T_1(\lambda) + T_2^{1/2}(\lambda)\\
		\Delta_2(\lambda) = T_1(\lambda) - T_2^{1/2}(\lambda)
		\end{gather}
	\end{subequations}
	For $q_0 \in L_0^2(\bT)$, the function $\Delta(\lambda) = T_1(\lambda) + T_2^{1/2}(\lambda)$ is analytic on the two sheeted Riemann surface $\cR$ and the branches $\Delta_k$ of $\Delta$ have the forms 
	
	\begin{equation}
	\Delta_k(\lambda) = 
	\frac{1}{2}\big(\tau_k(\lambda)+\tau_k^{-1}(\lambda)\big)
	\end{equation}
	for $ \lambda \in \cR_k$ with $k = 1,2$. For the special case $q_0 \equiv 0$, the corresponding functions are characterized by
	\begin{equation}
	\Delta_1^0(\lambda) = \cosh(\lambda^{1/2}), \qquad \Delta_2^0(\lambda) = \cos(\lambda^{1/2})
	\end{equation} 
 	For the operator $\cH^{\text{per}}$ the Lyapunov function $\Delta_1$ is increasing and $\Delta_2$ is bounded on the real line at high energy-level (large $\lambda$ values). The Lyapunov function for the operator $\cH^{\text{per}}$ defines the band structure of the spectrum, but it is an analytic function on a 2-sheeted Riemann surface. The qualitative behavior of the Lyapunov function for identically vanishing and small potentials are shown in Figure \ref{fig:bandStructure}. 
 	\begin{figure}[ht]
 		\centering
 		\includegraphics[width=0.975\textwidth]{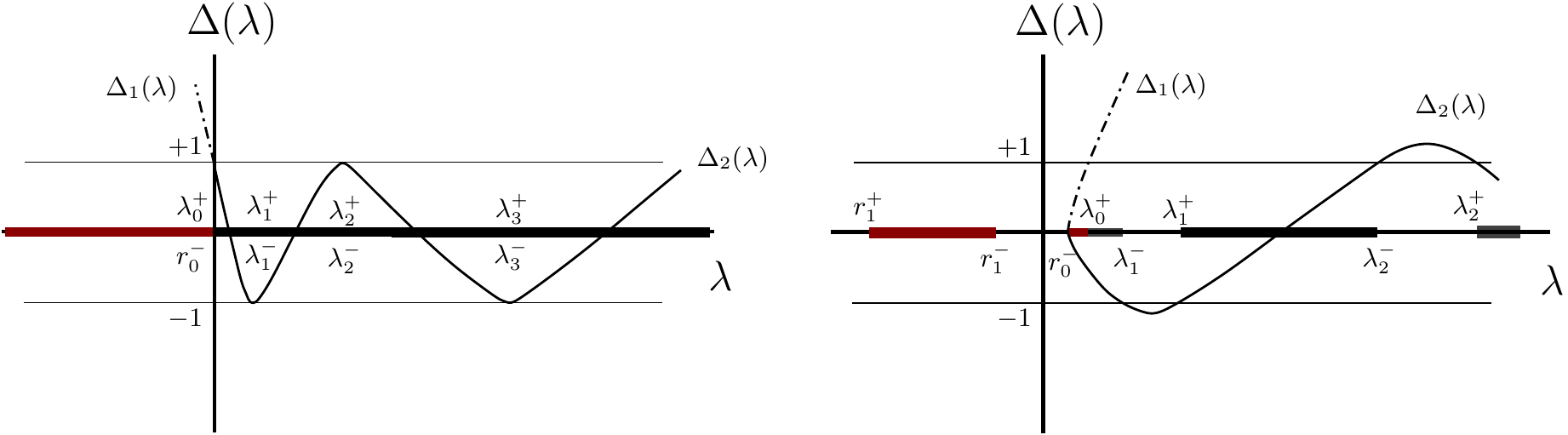}
 		\caption{The function $\Delta$ for the zero potential and a small potential $q_0$.}
 		\label{fig:bandStructure}
 	\end{figure}
	\begin{remark}
		In the case of Hill's operator, the monodromy matrix has exactly 2 eigenvalues $\tau$ and $\tau^{-1}$. The Lyapunov function $\frac{1}{2}(\tau+\tau^{-1})$ is an entire function of the spectral parameter. It defines the band structure of the spectrum, see \cite{K16} for detailed discussions. 
	\end{remark}
	\begin{thm}\emph{\textbf{(spectra of $\cH^{\text{per}}$ \cite{BL04, BK05, BK10})}}
		\label{summaryRefResults}
		Let $\Delta_1(\lambda)$ and $\Delta_2(\lambda)$ as defined in \eqref{eq:Delta12}, then for eigenvalue problem \eqref{eq:rigrnValRealPer} following results hold
		\begin{itemize}
			\item[(i)] The spectrum $\sigma(\cH^{\text{per}})$ of $\cH^{\text{per}}$ is purely absolutely continuous. 
			\item[(ii)] $\lambda \in \sigma(\cH^{\text{per}})$ iff $\Delta_k(\lambda) \in [-1,1]$ for some $k=1,2$. If $\lambda \in \sigma(\cH^{\text{per}})$, then $T_2(\lambda) \geq 0$. 
			\item[(iii)] There exists an integer $n_0 \in \bN_0$ such that for all $n \geq n_0$ 
			\begin{equation}
			\lambda_n^- \leq \lambda_n^+ \leq \lambda_{n+1}^- \leq \lambda_{n+1}^+ \leq \lambda_{n+2}^- \leq \lambda_{n+2}^+ \leq \cdots
			\end{equation}
			where the intervals $[\lambda_{n}^+,\lambda_{n+1}^-]$ are spectral bands of multiplicity 2 in $(\lambda_{n}^-,\lambda_{n+1}^-)$, and the intervals $(\lambda_{n}^-,\lambda_{n}^+)$ are gaps. 
			\item[(iv)] Each gap $G_n =(E_n^-,E_n^+)$ for $n \in \mathbb{N}$ is a bounded interval and $E_n^\pm$ are either periodic (anti-periodic) eigenvalues or resonance point, namely, real branch point of $\Delta_k$ for some $k = 1,2$ which is a zero of $T_2(\lambda)$. 
			\item[(v)] Any $\lambda \in \sigma(\cH^{\text{per}})$ on an interval $S \subset \bR$ has multiplicity 4 iff $-1 < \Delta_k(\lambda) < 1$ for all $k=1,2$ and $\lambda \in S$, except for finite number of points. 
			\item[(vi)] Any $\lambda \in \sigma(\cH^{\text{per}})$ on an interval $S \subset \bR$ has multiplicity 2 iff $-1 < \Delta_1(\lambda) < 1$, $\Delta_2(\lambda) \in \bR \setminus [-1,1]$ or $-1 < \Delta_2(\lambda) < 1$, $\Delta_1(\lambda) \in \bR \setminus [-1,1]$ for all $\lambda \in S$, except for finite number of points. 
			\item[(vii)] Let $\Delta_k$ be real analytic on some interval $I \subset \bR$ and $-1<\Delta_k(\lambda) < 1$ for any $\lambda \in I$ for some $k = 1,2$. Then $\Delta_k'(\lambda) \not=0$ for $\lambda \in I$ (monotonicity).   
			\item[(viii)] The dispersion relation for $\cH^{\text{per}}$ is given by 
			\begin{equation}
			\Delta_{1,2}(\lambda) := T_1(\lambda) \pm T_2^{1/2}(\lambda) = \cos(\theta),
			\end{equation}
			where $\theta$ is the one-dimensional quasimomentum. 
		\end{itemize}
	\end{thm}
	\section{Spectra of Hexagonal Lattice Hamiltonian}
	\label{sec:SHLH}
	In this section our aim is to adapt and characterize the spectrum $\sigma(\cH)$ of operator $\cH$ defined in Theorem \ref{MainTheorem} on (graphene like) hexagonal lattices. Due to positiveness and self-adjointness of this operator, its spectrum is real and positive. Let $\lambda \in \sigma(\cH)$ with $\lambda > 0$ be an eigenvalue of $\cH$ with associated eigenfunction $(u_e)_{e\in E} \in \cD(\cH)$. Note that since $a_e$ in \eqref{eq:diffSystem} is a positive constant and identical over the hexagonal lattice, we assume it is identically one. Then $u_e(x)$ satisfies on each edge $e \in E$
	\begin{equation}
		\cH u_e(x) = u_e''''(x)  + q(x) u_e(x) = 
		\lambda  u_e(x) 
	\end{equation}
	Let 
	$\delta_0 :=2\pi/3$ and define angles 
	\begin{equation}
	\label{eq:anglePerturb}
		\delta_c\up{\eps} := \delta_0 + c \eps 
	\end{equation}
	For $c_1 \in [-1,1]$ to be an arbitrary parameter and $c_2 := -(1+c_1)$, the eigenfunction $(u_e)_{e\in E} \in \cD(\cH)$ corresponding to the $\eps$-perturbed lattice at each vertex $\vt$ satisfy (see Theorem \ref{MainTheorem}) primary vertex conditions 
	\begin{subequations}
		\label{eq:PCond}
		\begin{gather}
			\label{eq:P1}
			u_1(\vt) = u_2(\vt)  = u_3(\vt) \\
			\label{eq:P2}
			\sin(\delta_{1}\up{\eps}) u_1'(\vt) +\sin(\delta_{c_1}\up{\eps}) u_2'(\vt) +\sin(\delta_{c_2}\up{\eps}) u_3'(\vt) = 0.
		\end{gather}
	\end{subequations}
	along with their conjugate ones
	\begin{subequations}
		\label{eq:SCond}
		\begin{gather}
			\label{eq:S1}
			\sin^{-1}(\delta_{1}\up{\eps})u_1''(\vt) = \sin^{-1}(\delta_{c_1}\up{\eps})u_2''(\vt)  = \sin^{-1}(\delta_{c_2}\up{\eps})u_3''(\vt), \\
			\label{eq:S2}
			u_1'''(\vt) + u_2'''(\vt) + u_3'''(\vt) = 0.
		\end{gather}
	\end{subequations}
	We stress out here that the result in \eqref{eq:S1} is obtained by setting $\ell, \ell' = 2,3$ in \eqref{eq:secondBcThmRota} and using the fact that $\vec i_3 \cdot \vec j_2 = - \vec i_2 \cdot \vec j_3$. Moreover for graphene and its $\eps$-perturbed angles, conditions above are well-defined, we refer reader to \cite{KKU15} for discussion about special cases e.g. when $\delta_1\up{\eps} = \pi$. 
	
	Density of states is determined by the dispersion relation, and thus when the latter is known, the former can be determined as well \cite{K16}. Thereby, we apply now the standard Floquet-Bloch theory with respect to the $\mathbb{Z}^2$-action that we specified before. This reduces the study of the Hamiltonian $\cH$ to the study of the family of Bloch Hamiltonians $\cH^{\Theta}$ acting in $L^2(W)$ for the values of the \textit{quasimomentum} $\Theta$ in the (first) Brillouin zone $[-\pi, \pi]^2$. Here the Bloch Hamiltonian $\cH^{\Theta}$ acts the same way $\cH$ does, but it is applied to a different space of functions. Each function $u = \{u_e\}_{e \in E}$ in the domain of $H^{\Theta}$ must belong to the Sobolev space $H^4(e)$ on each edge $e$ and satisfy the vertex conditions \eqref{eq:PCond}-\eqref{eq:SCond}, as well as the cyclic conditions (Floquet-Bloch conditions)
	\begin{equation}
		\label{eq:FBThm}
		\begin{split}
			u(x + n_1\vec{b}_1 + n_2\vec{b}_2) = e^{i\vec n\cdot\Theta}u(x) = e^{i(n_1\theta_1+n_2\theta_2)}u(x) 
		\end{split}
	\end{equation}
	for any $x$ in the fundamental domain $W$, vector $\vec n = (n_1,n_2) \in \mathbb{Z}^2$, and quasimomentum $\Theta = (\theta_1,\theta_2) \in [-\pi, \pi]^2$, see Figure \ref{fig:fundDomain}. Due to the condition \eqref{eq:FBThm}, function $u$ is uniquely determined by its restriction to the fundamental domain $W$. Then conditions \eqref{eq:PCond}-\eqref{eq:SCond} at the central vertex, i.e. $x = 0$ become
	\begin{subequations}
		\begin{gather}
			\label{eq:cond1At0}
			u_1(0)  = u_2(0) = u_3(0) =:A \\
			\label{eq:cond2At0}
			\sin(\delta_{1}\up{\eps}) u_1'(0) +\sin(\delta_{c_1}\up{\eps}) u_2'(0) +\sin(\delta_{c_2}\up{\eps}) u_3'(0) = 0\\
			\label{eq:cond3At0}
			\sin^{-1}(\delta_{1}\up{\eps})u_1''(0) = \sin^{-1}(\delta_{c_1}\up{\eps})u_2''(0)  = \sin^{-1}(\delta_{c_2}\up{\eps})u_3''(0)=: B \\
			\label{eq:cond4At0}
			u_1'''(0) + u_2'''(0) + u_3'''(0) = 0.
		\end{gather}
	\end{subequations}
	Similarly at other end vertex of edge $e_1$, i.e. $x = 1$ we have
	\begin{subequations}
		\begin{gather}
			\label{eq:cond1At1}
			u_1(1)  = u_2(1)e^{i \theta_1} = u_3(1)e^{i \theta_2}  =:C \\
			\label{eq:cond2At1}
			\sin(\delta_{1}\up{\eps}) u_1'(1) +\sin(\delta_{c_1}\up{\eps}) u_2'(1)e^{i \theta_1}  +\sin(\delta_{c_2}\up{\eps}) u_3'(1)e^{i \theta_2}  = 0\\
			\label{eq:cond3At1}
			\sin^{-1}(\delta_{1}\up{\eps})u_1''(0) = \sin^{-1}(\delta_{c_1}\up{\eps})u_2''(0)e^{i \theta_1} = \sin^{-1}(\delta_{c_2}\up{\eps})u_3''(0) e^{i \theta_2} =: D \\
			\label{eq:cond4At1}
			u_1'''(1) + u_2'''(1)e^{i \theta_1} + u_3'''(1)e^{i \theta_2}= 0.
		\end{gather}
	\end{subequations}
	By standard arguments, $\cH^{\Theta}$ has purely discrete spectrum $\sigma(\cH^{\Theta}) = \{\lambda_k(\Theta)\}_{k \in \mathbb{N}}$. The graph of the multiple valued function $\Theta \mapsto \{\lambda_k(\Theta)\}$ is known as the \textit{dispersion relation}, or \textit{Bloch variety} of the operator $\cH$. It is known \cite{K16} that the range of this function is the spectrum of $\cH$:
	\begin{equation}
		\sigma(\cH) = \bigcup_{\Theta \in [-\pi,\pi]^2}\sigma(\cH^{\Theta})
	\end{equation}
	Our goal now is the determination of the spectrum of $\cH^{\Theta}$ 
	and thus the dispersion relation of $\cH$. In order to determine this spectrum, we have to solve the eigenvalue problems
	\begin{equation}\label{2ndOrderOp}
		\cH^{\Theta}u(x) = \lambda u(x) 
	\end{equation}
	for $\lambda \in \mathbb{R}$ and non-trivial functions $u_e(x) \in L_e^2(W)$ with the above boundary conditions. Let us denote by $\Sigma^D$ the spectrum of the operator 
	\begin{equation}\label{eigenvalueEquation}
		\cH u(x) = au''''(x) +q(x)u(x)
	\end{equation}
	on interval $(0,1)$ with boundary conditions 
	\begin{equation}
		\label{eq:SigmaDboundary}
		u(0) = 0, \quad  u''(0) = 0, \quad  
		u(1) = 0, \quad  u''(1)  = 0.
	\end{equation}
	If $\lambda \notin \Sigma^D$, then there exist four linearly independent solutions $\phi_1, \phi_2,  \phi_3$ and $ \phi_4$ (depending on $\lambda$) of \eqref{eigenvalueEquation}
	on $(0,1)$ such that 
	\begin{equation}
	\label{eq:phiIndSol}
	\begin{split}
	\phi_1(0) = 1, \qquad \phi_1''(0)  = 0, \qquad  \phi_1(1) = 0, \qquad  \phi_1''(1) = 0,  \\
	\phi_2(0) = 0, \qquad \phi_2''(0)  = 1, \qquad  \phi_2(1) = 0, \qquad  \phi_2''(1) = 0,  \\
	\phi_3(0) = 0, \qquad \phi_3''(0)  = 0, \qquad  \phi_3(1) = 1, \qquad  \phi_3''(1) = 0,  \\
	\phi_4(0) = 0, \qquad \phi_4''(0)  = 0, \qquad  \phi_4(1) = 0, \qquad  \phi_4''(1) = 1.
	\end{split}
	\end{equation}
	For example, if $q \equiv 0$ and $\lambda > 0$, then we have $\lambda \not \in \Sigma^D$ if and only if $\lambda^{1/4} \not \in \pi \bZ$. If $\lambda \not \in \Sigma^D$, then
	\begin{equation*}
	\begin{split}
	\phi_1(x) = \frac{1}{2}&\big(\cos(\lambda^{1/4} x)+\cosh(\lambda^{1/4} x)+\cot(\lambda^{1/4} )\cos(\lambda^{1/4} x)-\coth(\lambda^{1/4})\cosh(\lambda^{1/4} x)\big)
	\end{split}
	\end{equation*} 
	and so on. We will assume that the functions $\phi_k$ are lifted to each of the edges in $W$,
	using the identifications of these edges with the segment $[0, 1]$ described above. Abusing notations, we will use the same names $\phi_k$ for the lifted functions. 
	For $\lambda \not \in \Sigma^D$
	one can use \eqref{eq:phiIndSol} to represent any solution $u$ of \eqref{2ndOrderOp} from the domain of $\cH\up{\Theta}$ on each edge in $W$ as follows:
	\begin{equation}
		\label{eq:funcVs}
		\begin{split}
			&u_1(x) = A  \phi_1(x) +  B \sin(\delta_1\up{\eps}) \phi_2(x) + C \phi_3(x)e^{-i\theta_0} +
			D \sin(\delta_1\up{\eps}) \phi_4(x)e^{-i\theta_0} \\
			&u_2(x) = A  \phi_1(x) + B \sin(\delta_{c_1}\up{\eps}) \phi_2(x) + C \phi_3(x) e^{-i\theta_1} +
			D \sin(\delta_{c_1}\up{\eps}) \phi_4(x) e^{-i \theta_1}\\
			&u_3(x) = A  \phi_1(x) +  B \sin(\delta_{c_2}\up{\eps}) \phi_2(x) + C \phi_3(x) e^{-i\theta_2} +
			D \sin(\delta_{c_2}\up{\eps}) \phi_4(x) e^{-i \theta_2}
		\end{split}
	\end{equation}
	with $\theta_0 = 0$. Next, let us introduce (Wronskian) operator $\cW : L^2[0,1] \times L^2[0,1] \rightarrow \bC$, defined as
	\begin{equation}
		\cW_x(u_1,u_2) := u_1'''(x)u_2(x) - u_1''(x)u_2'(x) + u_1'(x)u_2''(x) - u_1(x)u_2'''(x)
	\end{equation}
    for $x \in [0,1]$.
	Then for fourth-order Hamiltonian $\cH$ we get
	\begin{equation}
		\label{eq:WronskianSym}
		u_2(x) \cH u_1(x) - u_1(x)\cH u_2(x) = \big( \cW_1(u_1,u_2) - \cW_0(u_1,u_2)\big)'.
	\end{equation}
	If $u_1$ and $u_2$ solves $\cH u =\lambda u$, then $\cW_1(u_1,u_2) - \cW_0(u_1,u_2)$ is a constant. In the next lemma we use $\mathcal{W}$ to show some identities of $\phi_k$ at the end points.  
	\begin{lem}
		\label{symResultPhi}
		Applying symmetry property of operator $\cH$ acting on interval $(0,1)$ we get
		\begin{equation*}
			\begin{split}
				\phi_3'(1) = -\phi_1'(0), \quad \phi_3'(0) = -\phi_1'(1), \quad \phi_3'''(1) = -\phi_1'''(0), \quad 
				\phi_3'''(0) = -\phi_1'''(1), \quad \phi_2'''(0) = \phi_1'(0), \\
				\phi_4'(1) = -\phi_2'(0), \quad \phi_4'(0) = -\phi_2'(1), \quad \phi_4'''(1) = -\phi_2'''(0), \quad
				\phi_4'''(0) = -\phi_2'''(1), \quad \phi_2'''(1) = \phi_1'(1).
			\end{split}
		\end{equation*}
	\end{lem}
	\begin{proof}[\normalfont \textbf{Proof of Lemma~\ref{symResultPhi}}] 
		Proof of this lemma is based on \eqref{eq:WronskianSym}. In fact, for $n,m \in \{1,2,3,4\}$ and $n \not=m$, let $\phi_n(x)$ and $ \phi_m(x)$ be two independent solutions of eigenvalue problem 
		\begin{equation}
			\cH u(x)= u''''(x) + q(x)u(x) = \lambda u(x)
		\end{equation}
		on $(0,1)$ satisfying boundary conditions \eqref{eq:phiIndSol}. Now observe that
		\begin{equation}\label{Weq1}
			\phi_m(x) \cH\phi_n(x) - \phi_n(x) \cH \phi_m(x) = \phi_m(x)\lambda \phi_n(x) - \phi_n(x) \lambda \phi_m(x) = 0.
		\end{equation}
		However by \eqref{eq:WronskianSym} 
		\begin{equation}\label{Weq2}
			\phi_m(x) \cH\phi_n(x) - \phi_n(x) \cH \phi_m(x)  = \big( \cW_1(\phi_n,\phi_m) - \cW_0(\phi_n,\phi_m)\big)'.
		\end{equation}
		For a constant $\text{c}$, \eqref{Weq1} and \eqref{Weq2} then imply that $ \cW_1(\phi_n,\phi_m) - \cW_0(\phi_n,\phi_m) = \text{c}$. For any choice of $n \not= m$, observe that the boundary conditions in \eqref{eq:phiIndSol} implies $c = 0$, i.e. 
		\begin{equation}
			\cW_1(\phi_n,\phi_m) = \cW_0(\phi_n,\phi_m). 
		\end{equation}
		Finally, applying properties of $\phi_n$ from \eqref{eq:phiIndSol}, one concludes the desired result. As an example setting $(n,m)=(1,3)$ and using the property that the only non-zero terms are $\phi_1(0)$ and $\phi_3(1)$, then
		\begin{equation*}
			\phi_3'''(0) = - \phi_1'''(0).
		\end{equation*}
		Similar conclusions can be made to derive the desired relations stated in the lemma. 
	\end{proof}
	\begin{defn}
		For $k \in \bN_0 = \bN \cup \{0\}$ and arbitrary $\eps > 0$, let
		\begin{equation}
			\label{eq:sThetaFunc}
			S_k\up{\eps}(\Theta) :=\sin^{-k}(\delta_0)\big(\sin^k(\delta_{1}\up{\eps}) + \sin^k(\delta_{c_1}\up{\eps})e^{-i\theta_1} + \sin^k(\delta_{c_2}\up{\eps})e^{-i\theta_2} \big),
		\end{equation}
		where $\Theta \in [-\pi,\pi]^2$, $\delta_0 = 2\pi/3$, $c_1 \in [-1,1]$ and $c_2 = -(1+c_1)$.
	\end{defn}
	Let us introduce scaled version of  $\tilde B := \sin(\delta_0) B$ and $\tilde D := \sin(\delta_0) D$ stated in \eqref{eq:cond3At0} and \eqref{eq:cond3At1} respectively. Application of function $u_i$s defined in \eqref{eq:funcVs} in vertex conditions \eqref{eq:cond2At0}, \eqref{eq:cond4At0} and  \eqref{eq:cond2At1}, \eqref{eq:cond4At1} reduces the problem to find the vector $\vec \xi := (A~\tilde B ~C ~\tilde D)^T$ satisfying 
	\begin{equation}
		\label{eq:linearSystem1}
		\mathbb{M}_\eps \vec \xi = 
		\begin{pmatrix}
			\hspace{4mm}A_0(\eps) & -A_1(\eps)  \\
			-\widetilde A_1(\eps) & \hspace{2.5mm}A_0(\eps) 
		\end{pmatrix}
		\vec \xi = 0.
	\end{equation}
	The block components of matrix $\mathbb{M}_{\varepsilon}$ are written in terms of quasimomentum and solutions $\phi_k$s and have forms 
	\begin{equation*}
		A_0(\eps) := 
		\begin{pmatrix}
			S_1\up{\eps}(0) \phi_1'(0) & S_2\up{\eps}(0) \phi_2'(0)\\
			S_0\up{\eps}(0) \phi_1'''(0) & S_1\up{\eps}(0) \phi_2'''(0)
		\end{pmatrix}
		,\quad 
		A_1(\eps) := 
		\begin{pmatrix}
			S_1\up{\eps}(\Theta) \phi_1'(1) & S_2\up{\eps}(\Theta) \phi_2'(1)\\
			S_0\up{\eps}(\Theta) \phi_1'''(1) & S_1\up{\eps}(\Theta) \phi_2'''(1)
		\end{pmatrix}
	\end{equation*}
	and 
	\begin{equation*}
		\widetilde  A_1(\eps) := -
		\begin{pmatrix}
			\overline{S_1\up{\eps}(\Theta)} \phi_1'(1) & \overline{S_2\up{\eps}(\Theta)} \phi_2'(1)\\
			\overline{S_0\up{\eps}(\Theta)} \phi_1'''(1) &\overline{S_1\up{\eps}(\Theta)} \phi_2'''(1)
		\end{pmatrix}.
	\end{equation*}
	Clearly a non-trivial solution exists if matrix $\mathbb{M}_\eps(\lambda)$ is singular, stated formally as the following proposition.
	\begin{prop}
		\label{detProp}
		If $\lambda \not \in \Sigma^D$, then $\lambda$ is in spectrum of the hexagonal elastic lattice's Hamiltonian $\cH$ if and only if there is $\Theta \in [-\pi,\pi]^2$ such that 
		\begin{equation}
			\label{eq:MepsCond}
			\det\big(\mathbb{M}_{\eps} (\lambda)\big) = 0.
		\end{equation}
	\end{prop}
	The result in Proposition \ref{detProp} can be (numerically) investigated directly, but we will split the discussions into two parts. In the following section we will state theoretical results for the case $\eps = 0$, namely graphene lattice. In Section \ref{sec:APH}, extension of these results will be presented for perturbed angles by applying tools from perturbation analysis. 
	\section{Graphene Hamiltonian}
	\label{sec:GH}
	In this section we discuss the outcome of results from the previous section for the special case of $\eps=0$. In this case for any $k \in \bN_0$ let
	\begin{equation}
		\label{eq:s0}
		s_0(\Theta) := S_k\up{0}(\Theta) = 1 + e^{-i\theta_1} + e^{-i\theta_2}.
	\end{equation}
	Application of $s_0(\Theta)$ reduces the block matrix components defined in \eqref{eq:linearSystem1} into splitted forms  
	\begin{equation}
		A_0(\lambda)= s_0(0) \Phi_0(0), \qquad A_1(\lambda)= -s_0(\Theta) \Phi_0(1), \qquad \widetilde A_1(\lambda)= -\overline{s_0(\Theta)} \Phi_0(1),
	\end{equation}
	in which matrices $\Phi_0(0)$ and $\Phi_0(1)$ are of the form
	\begin{equation}
		\label{eq:simplifiedAs}
		\Phi_0(0) :=
		\begin{pmatrix}
			\phi_1'(0) & \phi_2'(0)\\
			\phi_1'''(0) &  \phi_2'''(0)
		\end{pmatrix}
		,\qquad 
		\Phi_0(1) := 
		\begin{pmatrix}
			\phi_1'(1) &\phi_2'(1)\\
			\phi_1'''(1) & \phi_2'''(1)
		\end{pmatrix}.
	\end{equation}
	\begin{lem}
		\label{A1NonSingular}
		The matrix $\Phi_0(1)$ defined in \eqref{eq:simplifiedAs} is non-singular. 
	\end{lem}
	\begin{proof}[\normalfont \textbf{Proof of Lemma~\ref{A1NonSingular}}] 
		By contradiction, let's assume $\Phi_0(1)$ is singular, which by application of relations in Lemma \eqref{symResultPhi} reduces to the condition 
		\begin{equation}
			\label{eq:detPhi}
			\det(\Phi_0(1)) = \phi_1'(1)\phi_2'''(1) - \phi_2'(1)\phi_1'''(1) = \phi_3'(0)\phi_4'''(0) - \phi_4'(0)\phi_3'''(0) = 0.
		\end{equation}
		Using the fact that $ \phi_4'''(0) = \phi_3'(0)$, \eqref{eq:detPhi} implies (at least) one of the following conditions is true:
		\begin{itemize}
			\item[(i)] $\phi_3'(0) = 0 \quad \& \quad \phi_3'''(0) = 0$, 
			\item[(ii)] $\phi_3'(0) = 0 \quad \& \quad \phi_4'(0) = 0$, 
			\item[(iii)] $\phi_3'(0) \neq 0\quad \& \quad\phi_3'''(0) \neq 0\quad \& \quad\phi_4'(0) \neq 0$. 
		\end{itemize}
		Recall fundamental solutions $g_k$ from Subsection \ref{4thorderHill}. Then using the fact that $\phi_3$ can be represented as a linear combination of them in the form 
		\begin{equation}
			\label{eq:repG}
			\phi_3(x) = b_1 g_1(x) +  b_2 g_2(x) +  b_3 g_3(x) +  b_4 g_4(x) 
		\end{equation}
		along with the conditions $\phi_3(0) = 0$, $\phi_3''(0) = 0$ implies that item (i) turns to $\phi_3 \equiv 0$, which is a contradiction. A similar discussion holds to show that item (ii) above results in $\phi_4 \equiv 0$. Now, considering the last case above, let us introduce
		\begin{equation}
			r := \frac{\phi_3'(0)}{\phi_3'''(0)} = \frac{\phi_4'(0)}{\phi_3'(0)}.
		\end{equation}
		Then obviously by our assumption $r \not= 0$. Utilizing representation \eqref{eq:repG} and a similar representation for $\phi_4$, one gets 
		\begin{align}
			\phi_3(x) = \phi_3'(0)g_2(x) + \frac{\phi_3'(0)}{r}g_4(x), \qquad 
			\phi_4(x) = r\phi_3'(0)g_2(x) + \phi_3'(0)g_4(x).
		\end{align}
		Comparing these two representations implies that $\phi_4 = r\phi_3$, which is a contradiction, since  by our assumption $\phi_3$ and $\phi_4$ are linearly independent solutions. This proves the desired claim of the non-singularity of matrix $\Phi_0(1)$.
	\end{proof}
	Applying the fact that $s_0(0) = 3$ along with Lemma \ref{A1NonSingular} reduces condition \eqref{eq:MepsCond} in Proposition \ref{detProp} to 
	\begin{equation}
		\label{eq:linearSystem2}
		\det\Big( \Lambda_0^2(\lambda) -  \frac{|s_0(\Theta)|^2}{9} \bI_2 \Big) = 0, 
	\end{equation}
	where $	\Lambda_0(\lambda) := \Phi_0^{-1}(1)\Phi_0(0)$. Then for the graphene lattice we proved the following result.
	\begin{prop}
		\label{detLem}
		If $\lambda \not \in \Sigma^D$, then $\lambda$ is in spectrum of the hexagonal elastic lattice's Hamiltonian $\cH$ if and only if there is $\Theta \in [-\pi,\pi]^2$ such that 
		\begin{equation}
			\det\Big(\Lambda_0(\lambda) -\frac{|s_0(\Theta)|}{3} \bI_2\Big) \det\Big(\Lambda_0(\lambda) +\frac{|s_0(\Theta)|}{3} \bI_2\Big) = 0.
		\end{equation}
		In other words $|s_0(\Theta)|/3$ is a root of the characteristic polynomial for $\Lambda_0(\lambda)$ or $-\Lambda_0(\lambda)$ matrices, i.e. a root of
		\begin{equation}\label{polydispersion}
			\cP(z;\lambda) = \big(z^2 - \tr(\Lambda_0(\lambda))z + \det(\Lambda_0(\lambda)\big)\big)\big(z^2 + \tr(\Lambda_0(\lambda))z + \det(\Lambda_0(\lambda))\big).
		\end{equation}
	\end{prop}
	Proposition \ref{detLem}, in particular, says that in order to find the spectrum of $\cH$, we
	need to calculate the range of $|s_0(\Theta)|$ on $[-\pi,\pi]^2$. This function is identical to the one reported for the second order {S}chr{\"o}dinger operator on graphene \cite{KP07}. In summary, $|s_0(\Theta)|$ has range $[0,3]$, its maximum is attained at $(0,0)$ and minimum at $\pm (\delta_0,-\delta_0)$. These properties are based on a simple observation that
	\begin{equation}
		|s_0(\Theta)|^2 = |1+e^{i \theta_1} + e^{i\theta_2}|^2
	\end{equation}
	with range $[0,9]$. See Figure \ref{Res0Bars1} (left) for a plot of the level curves of this function.
	\subsection{Dispersion Relation via Fundamental Solutions}
	Next, we interpret the functions $\phi_k$ and hence matrix $\Lambda_0$ in terms of the original potential $q_0$ on $[0,1]$. To this end, let us extend $q_0$ periodically to the real line and consider operator $\cH^{\text{per}}$ on $\mathbb{R}$, defined in preliminary section as
	\begin{equation}
		\label{perH}
		\cH^{\text{per}}u(x) = u''''(x) + q_0(x) u(x)
	\end{equation} 
	with the periodic potential extended from $q_0$. Note that with the abuse of notation we maintain the notation $q_0$ for the extended potential. Fundamental solutions $\{g_k(x)\}_{k=1}^4$ of $\cH^{\text{per}}$ satisfy for $j, k \in \{1,\ldots,4\}$ conditions
	\begin{equation}
		g_k^{(j-1)}(0)	 = \delta_{j k}.
	\end{equation}
	Thereby, monodromy matrix $M(\lambda)$ defined through \eqref{eq:monodMatrix} shifts by the period along the solutions of \eqref{perH}, i.e. 
	\begin{equation*}
		\begin{pmatrix}
			u(1) \\
			u'(1) \\ 
			u''(1) \\
			u'''(1) 
		\end{pmatrix} = 
		\begin{pmatrix}
			g_1(1) & g_2(1) & g_3(1) & g_4(1) \\
			g_1'(1) & g_2'(1) & g_3'(1) & g_4'(1) \\
			g_1''(1) & g_2''(1) & g_3''(1) & g_4''(1) \\
			g_1'''(1) & g_2'''(1) & g_3'''(1) & g_4'''(1) 
		\end{pmatrix}
		\begin{pmatrix}
			u(0) \\
			u'(0) \\ 
			u''(0) \\
			u'''(0) 
		\end{pmatrix}
	\end{equation*}
	The $4\times4$ matrix valued function $\lambda \mapsto M(\lambda)$ is entire, see the preliminaries Section \ref{sec:Preliminaries} and references therein for more detailed discussions. Since our goal is to obtain the dispersion relation of the operator $\cH$, next we derive relations among $g_k$ and $\phi_k$. For simplicity let us introduce the following notation:
	\begin{equation}\label{detNotation}
		\cD(f,g) := f'(0)g'''(1) - g'(1)f'''(0).
	\end{equation}
	\begin{lem}\label{phitoFundSol}
		Fundamental solutions $\{g_k(x)\}_{k=1}^4$  of $\cH^{\text{per}}$ can be represented in terms of the functions $\phi_1$ and $\phi_2$ introduced in \eqref{eq:phiIndSol} as:
		\begin{align*}
			g_1(x) &= \phi_1(x) + \frac{1}{\det(\Phi_0(1))}\big(\cD(\phi_1,\phi_2)\phi_3(x) - \cD(\phi_1,\phi_1) \phi_4(x)\big),\\
			g_3(x) &= \phi_2(x) + \frac{1}{\det(\Phi_0(1))}\big(\cD(\phi_2,\phi_2)\phi_3(x) + \cD(\phi_1,\phi_2) \phi_4(x)\big),
		\end{align*}
		and moreover
		\begin{align*}
			g_2(x) &=  \frac{-1}{\det(\Phi_0(1))}\big(\phi_1'(1) \phi_3(x) - \phi_1'''(1)  \phi_4(x)\big),\\
			g_4(x) &=  \frac{1}{\det(\Phi_0(1))}\big(\phi_2'(1) \phi_3(x) - \phi_1'(1)  \phi_4(x)\big).
		\end{align*}
	\end{lem}
	\begin{proof}[\normalfont \textbf{Proof of Lemma~\ref{phitoFundSol}}] 
		Starting with the property that $\{\phi_k\}_{k=1}^4$ and $\{g_k\}_{k=1}^4$ solve the eigenvalue problem 
		\begin{equation}
			u''''(x) + q(x)u(x) = \lambda u(x)
		\end{equation}
		and the fact that these are linearly independent sets of solutions, then each $g_k$ can be represented in the form 
		\begin{equation}
			g_k(x) = a_k\phi_1(x) + b_k\phi_2(x) + c_k\phi_3(x) + d_k\phi_4(x).
		\end{equation}
		Applying properties of $\phi_k$ given in \eqref{eq:phiIndSol}, we observe that coefficients corresponding to $g_1$ satisfy
		\begin{align*}
			g_1(0) = 0 \quad \Rightarrow \quad a_1 = 1, \hspace{20mm}
			g_1''(0) = 0 \quad \Rightarrow \quad b_1 = 0.
		\end{align*}
		Moreover, the remaining conditions result in 
		\begin{equation}
		  g_1'(0) = 0 \quad \Rightarrow \quad g_1'(0) = \phi_1'(0) + c_1\phi_3'(0) + d_1\phi_4'(0)
			= \phi_1'(0) - c_1\phi_1'(1) - d_1\phi_2'(1) = 0  
		\end{equation}
		and
		\begin{equation}
		  	g_1'''(0) = 0 \quad \Rightarrow \quad g_1'''(0) = \phi_1'''(0) + c_1\phi_3'''(0) + d_1\phi_4'''(0)
			= \phi_1'(0) - c_1\phi_1'''(1) - d_1\phi_2'''(1) = 0.  
		\end{equation}
		Solving for $c_1$ and $d_1$, we get 
		\begin{equation}
			c_1 = \frac{\cD(\phi_1,\phi_2)}{\det(\Phi_0(1))}, \qquad 
			d_1 = - \frac{\cD(\phi_1,\phi_1)}{\det(\Phi_0(1))}.
		\end{equation}
		Similar discussions can be followed to obtain the coefficients corresponding to remaining $g_k$. This finishes the proof. 
	\end{proof}
	Symmetry of the potential $q_0$ brings additional properties on the fundamental solutions which are summarized in the following lemma. 
	\begin{lem}
		\label{CondofFundSol}
		Under symmetry property of potential $q_0$, the fundamental solutions satisfy
		\begin{alignat*}{3}
			g_1''(1) &= g_2'''(1), \qquad &&g_1'(1) = g_3'''(1), \qquad  &&g_1(1) = g_4'''(1) \\
			g_2'(1) &= g_4'''(1), &&g_2(1) = g_4''(1),&&g_3(1) = g_4'(1)
		\end{alignat*}
	\end{lem}
	\begin{proof}[\normalfont \textbf{Proof of Lemma~\ref{CondofFundSol}}] 
		Establishing this result is similar to the proof of Lemma \ref{symResultPhi} along with an application of symmetry of potential. 
	\end{proof}
	Next let us introduce matrix $\mathbb{G}_0(\lambda)$  
	\begin{equation}
		\mathbb{G}_0(\lambda) := 
		\begin{pmatrix}
			g_1(1) & g_3(1)  \\
			g_1''(1) & g_3''(1) 
		\end{pmatrix}.
	\end{equation}
	This matrix can be interpreted as an extension of the (scalar-valued) discriminant function $D(\lambda) = g_1(1)+g_2'(1)$ for the eigenvalue problem corresponding to the second order Schr{\"o}dinger operator \cite{KP07}. Putting all the observations above together allows us to derive the dispersion relation of $\cH$ stated formally in the following theorem. 
	\begin{thm}\emph{\textbf{(dispersion relation)}}
		\label{detProp2}
		The dispersion relation of the hexagonal elastic lattice's Hamiltonian $\cH$ consists of the variety 
		\begin{equation}
			\label{eq:DispRelation}
			\det\Big(\mathbb{G}_0^2(\lambda) - \frac{|s_0(\Theta)|^2}{9} \bI_2 \Big) = 0
		\end{equation}
		and the collection of flat branches $\lambda \in \Sigma^D$. 
	\end{thm}
	\begin{proof}[\normalfont \textbf{Proof of Theorem~\ref{detProp2}}] 
		Recalling the notation (\ref{detNotation}), then applying Lemma \ref{phitoFundSol} and \eqref{eq:phiIndSol}, the following identities are in order:
		\begin{align*}
			g_1(1) + g_4'''(1) &=  +2\frac{\cD(\phi_1,\phi_2)}{\det(\Phi_0(1))}, \qquad
			g_3(1) + g_4'(1) =  +2\frac{\cD(\phi_2,\phi_2)}{\det(\Phi_0(1))},\\
			g_1''(1) + g_2'''(1) &=  -2\frac{\cD(\phi_1,\phi_1)}{\det(\Phi_0(1))}, \qquad
			g_3''(1) + g_2'(1) =  -2\frac{\cD(\phi_2,\phi_1)}{\det(\Phi_0(1))}.
		\end{align*}
		Since $\phi_2'''(0) = \phi_1'(0)$ and $\phi_2'''(1) = \phi_1'(1)$, observe that right-hand sides of the above equations are the entries of $2\Lambda_0(\lambda)$ introduced in \eqref{eq:linearSystem2}. Therefore using Lemma \ref{CondofFundSol} one gets
		\begin{equation*}
			2\Lambda_0(\lambda) = \begin{pmatrix}
				g_1(1)+g_4'''(1) & g_3(1)+g_4'(1)\\
				g_1''(1)+g_2'''(1) & g_3''(1)+g_2'(1) 
			\end{pmatrix} = 
			\begin{pmatrix}
				2g_1(1) & 2g_3(1)  \\
				2g_1''(1) & 2g_3''(1) 
			\end{pmatrix} = 2\mathbb{G}_0(\lambda),
		\end{equation*}
		Combining the results from Proposition \ref{detLem} and Lemma \ref{LemmasigmaD} establishes the claimed result. 
	\end{proof}
	For specific purposes, e.g. reducibility of Fermi surface, it may be desirable to rephrase \eqref{eq:DispRelation} in terms of characteristic polynomials.
	\begin{remark}
		$\lambda$ is in the spectrum of the graphene Hamiltonian $\cH$ if and only if $\lambda \in \Sigma^{D}$ or $|s_0(\Theta)|/3$ is a root of the characteristic polynomial for $\mathbb{G}_0(\lambda)$ or $-\mathbb{G}_0(\lambda)$, i.e. $\lambda \in \Sigma^{D}$ or is a root of 
		\begin{equation}
			\cP(z;\lambda) = \Big(z^2 - \tr\big(\mathbb{G}_0(\lambda)\big)z + \det\big(\mathbb{G}_0(\lambda)\big)\Big) \Big(z^2 + \tr\big(\mathbb{G}_0(\lambda)\big)z + \det\big(\mathbb{G}_0(\lambda)\big)\Big) 
		\end{equation}
		 by equation \eqref{polydispersion}.
	\end{remark} 
	Noting that $\phi_2'''(1) = \phi_1'(1)$ we can also write the dispersion relation as follows: $\lambda$ is in the Floquet spectrum of $\cH$ if and only if 
	\begin{equation}
		\Big(\Delta_1(\lambda) \pm \frac{|s_0(\Theta)|}{3} \Big)\Big(\Delta_2(\lambda)  \pm \frac{|s_0(\Theta)|}{3} \Big) = 0
	\end{equation}
	or $\lambda \in \Sigma^D$, where $\Delta_{1,2}(\lambda)$ were defined in \eqref{eq:Delta12} and 
	\begin{equation}
		\label{eq:T1T2Def}
		T_1 = \frac{ \tr(\mathbb{G}_0)}{2}, \qquad  T_2 = \frac{ \tr^2(\mathbb{G}_0)}{4} - \det(\mathbb{G}_0).
	\end{equation}
	So far, we have been avoiding points of the Dirichlet spectrum $\Sigma^D$ of a single
	edge. We will now deal with exactly these points. The idea is to explicitly construct corresponding eigenfunctions as discussed in \cite{KP07}.
	\begin{lem}
		\label{LemmasigmaD}
		Each point $\lambda \in \Sigma^D$ is an eigenvalue of infinite multiplicity of the hexagonal elastic lattice’s Hamiltonian $\cH$ on graphene. The corresponding eigenspace is generated by simple
		loop states, i.e. by eigenfunctions which are supported on a single hexagon and vanish at the vertices.
	\end{lem}
	\begin{proof}[\normalfont \textbf{Proof of Lemma~\ref{LemmasigmaD}}] 
		Let us first show that each $\lambda \in \Sigma^D$ is an eigenvalue. Let $u$ be an eigenfunction of the operator $d^4/dx^4 + q_0(x)$ with the eigenvalue $\lambda$ and (Dirichlet type) boundary conditions on $[0,1]$ as stated in \eqref{eq:SigmaDboundary}. Note that $u(1-x)$ is also an eigenfunction with the same eigenvalue, since $q_0(x)$ is even. If $u(x)$ is neither even nor odd, then $u(x) - u(1-x)$ is an odd eigenfunction. For an odd eigenfunction, repeating it on each of the six edges of a hexagon and letting the eigenfunction to be zero on any other hexagon, we get an eigenfunction of the operator $\cH$. If $u$ is an even eigenfunction, then repeating it around the hexagon with an alternating sign and letting the eigenfunction to be zero on any other hexagon, we get an eigenfunction of the operator $\cH$. Therefore $\lambda \in \sigma_{\text{pp}}(\cH)$. We get the rest of the proof by following the arguments of Lemma 3.5 in \cite{KP07}.
	\end{proof}
	\begin{remark}
		Compared to {S}chr{\"o}dinger, the Dirichlet type boundary conditions for fourth-order operator may be a place to be cautious. Naturally, one may select vanishing boundary conditions in quadratic form as Dirichlet ones (this choice holds for second-order operator). However, here we defined $\Sigma^D$ as \eqref{eq:SigmaDboundary} to accommodate Floquet vertex conditions in \eqref{eq:cond1At0}, \eqref{eq:cond3At0} and so on. We refer interested reader to the Section \ref{sec:Outlook} for further discussion along this line. 
	\end{remark}
	\begin{exmp}\label{freeOperator}
		Let us consider the free operator, i.e. $q \equiv 0$. Setting $\mu := \sqrt[4]{\lambda}$ and using the convention 
		\begin{equation}
			C_{\mu}^{\pm}(x) = \cosh(\mu x) \pm \cos(\mu x), \qquad S_\mu^{\pm}(x) = \sinh(\mu x) \pm \sin(\mu x),
		\end{equation}
		then the fundamental solutions take the forms 
		\begin{equation}
			g_1(x) = \frac{1}{2} C_{\mu}^{+}(x), \quad g_2(x) = \frac{1}{2\mu} S_{\mu}^{+}(x), \quad 
			g_3(x) = \frac{1}{2\mu^2} C_{\mu}^{-}(x), \quad g_4(x) = \frac{1}{2\mu^3} S_{\mu}^{-}(x),
		\end{equation}
		and hence 
		\begin{equation*}
			\mathbb{G}_0(\lambda) =
			\begin{pmatrix}
				g_1(1) & g_3(1)\\
				g_1''(1) & g_3''(1) 
			\end{pmatrix} =	\frac{1}{2}
			\begin{pmatrix}
				\hspace{4mm}C_{\mu}^{+}(1) &  \mu^{-2}C_{\mu}^{-}(1) \\
				\mu^2  C_{\mu}^{-}(1) &  \hspace{6mm}C_{\mu}^{+}(1)
			\end{pmatrix}.
		\end{equation*}
		This then implies
		\begin{align*}
			\det\Big(\mathbb{G}_0(\lambda) \pm \frac{|s_0(\Theta)|}{3} \bI_2 \Big) = 
			\Big(\frac{|s_0(\Theta)|}{3}\Big)^2 \pm \tr(\mathbb{G}_0)\Big(\frac{|s_0(\Theta)|}{3}\Big) + \det(\mathbb{G}_0).
		\end{align*}
		Thereby, the dispersion relation is equivalent to 
		\begin{equation}
			\Big(\cos(\lambda^{1/4})\pm\frac{|s_0(\Theta)|}{3}\Big)\Big(\cosh(\lambda^{1/4})\pm\frac{|s_0(\Theta)|}{3}\Big) = 0
		\end{equation}
		Since $\cosh(x) \geq 1$ and by taking to account $|s_0(\Theta)| \leq 3$, the only root of the second factor happens at $\Theta = (0,0)$ and $\lambda = 0$ which also solve the first phrase. Therefore the dispersion relation for $q_0 \equiv 0$ reduces to 
		\begin{equation}
			\label{eq:dispSurfZeroPot}
			\cos(\lambda^{1/4}) = \pm \frac{|s_0(\Theta)|}{3}.
		\end{equation}
	\end{exmp}
\begin{figure}[ht]
	\centering
	\includegraphics[width=0.45\textwidth]{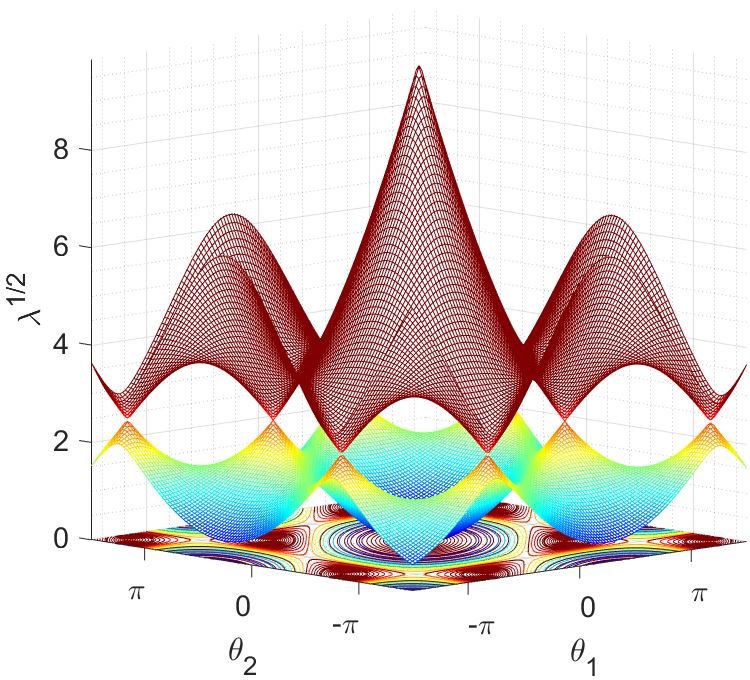}
	\caption{Dispersion relation for zero potential case, see \eqref{eq:dispSurfZeroPot}.}
	\label{fig:zeroPotDispersion}
\end{figure}
	\begin{remark}
		The dispersion relation of the second order Schr\"{o}dinger operator with the vanishing potential, i.e. $\cH\up{s} u(x) = -u''(x) $ on graphene has the form 
		\begin{equation}
			\cos(\lambda^{1/2}) = \pm \frac{|s_0(\Theta)|}{3},
		\end{equation}
		which is interestingly very similar to \eqref{eq:dispSurfZeroPot}. Therefore Example \ref{freeOperator} shows that the dispersion relation of the graphene Hamiltonian $\cH$ coincides with the one for the second order Schr\"{o}dinger operator on graphene $\cH\up{s}$ if the eigenvalue problems $\cH\up{s} u = \lambda u$ and $\cH u = \lambda^{1/2} u$ are considered. Figure \ref{fig:zeroPotDispersion} shows the plot of first two spectral sheets of the dispersion relation. 
	\end{remark}
	\subsection{The Spectra of Graphene Hamiltonian}
	This section is devoted on full description of spectra, conical singularities and Fermi surfaces corresponding to $\cH$ defined on graphene. 
	\begin{lem}
		\label{Dspecsubsetpapspec}
		As a set, $\Sigma^D$ belongs to the union of periodic and anti-periodic spectra of $\cH^{\text{per}}$.
	\end{lem}
	\begin{proof}[\normalfont \textbf{Proof of Lemma~\ref{Dspecsubsetpapspec}}] 
		Let $\lambda \in \Sigma^D$. Since the potential $q_0$ is even, if $u(x)$ is an eigenfunction, then $u(1-x)$ is also an eigenfunction. Therefore we can assume $u$ to be either even or odd. In case $u$ is odd, it satisfies the periodic boundary conditions, i.e.
		\begin{equation}\label{perbc}
			u(0)=u(1), \quad u'(0)=u'(1), \quad u''(0)=u''(1), \quad u'''(0)=u'''(1).
		\end{equation}
		On the other-hand for even $u$, it satisfies the anti-periodic boundary conditions
		\begin{equation}\label{aperbc}
			u(0)=-u(1), \quad u'(0)=-u'(1), \quad u''(0)=-u''(1), \quad u'''(0)=-u'''(1).
		\end{equation}
	\end{proof}
	We can now completely describe the spectral structure of the graphene operator $\cH$. 
	\begin{thm}
		\label{grapheneSpectrum}
		\emph{\textbf{(spectral description)}}
		\begin{itemize}
			\item[(i)] The singular continuous spectrum $\sigma_{\text{sc}}(\cH)$ is empty.
			\item[(ii)] The absolutely continuous spectrum $\sigma_{\text{ac}}(\cH)$ has band-gap structure and coincides as a set with the spectrum $\sigma(\cH^{\text{per}})$ of the 4-th order operator $\cH^{\text{per}}$ with potential $q_0$ periodically extended from $[0,1]$. Moreover, the absolutely continuous spectrum $\sigma_{\text{ac}}(\cH)$ has the representation
			\begin{equation}\label{abscontspectrum}
				\sigma_{\text{ac}}(\cH) = \big\{ \lambda \in \mathbb{R} ~|~ \Delta_k(\lambda)=[-1,1] \text{ for some } k=1,2 \big\}, 
			\end{equation}
			where $\Delta_{1,2}(\lambda) := \frac{1}{2}\big(\tr(\mathbb{G}_0(\lambda)) \pm \big(	\tr^2(\mathbb{G}_0(\lambda)) - 4\det(\mathbb{G}_0(\lambda))\big)^{1/2}\big)$
			\item[(iii)] The pure point spectrum $\sigma_{\text{pp}}(\cH)$ coincides with $\Sigma^D$ as a set and eventually belongs to the union of the edges of the spectral bands of $\sigma_{\text{ac}}(\cH)$.
		\end{itemize}
	\end{thm}
	\begin{proof}[\normalfont \textbf{Proof of Theorem~\ref{grapheneSpectrum}}] 
		Proof of the items above is based on the developed tools in this paper along with already-established results in our references. For item (i) observe that the singular continuous spectrum is empty, since $\cH$ is a self-adjoint elliptic operator (see e.g. Corollary 6.11 in \cite{K16}). Proof of (ii) is based on Theorem \ref{detProp2}, as we know that any $\lambda \notin \Sigma^D$ belongs to $\sigma(\cH)$ if and only if $|s_0(\Theta)|/3$ is a root of the characteristic polynomial for $D(\lambda)$ or $-D(\lambda)$, i.e. a root of 
		\begin{equation*}
			\cP(z;\lambda) := \big(z^2 - \tr(\mathbb{G}_0(\lambda))z + \det(\mathbb{G}_0(\lambda))\big) \big(z^2 + \tr(\mathbb{G}_0D(\lambda))z + \det(\mathbb{G}_0(\lambda))\big).
		\end{equation*}
		Since the range of $|s_0(\Theta)|$ is $[0,3]$, then $\cP(|s_0(\Theta)|/3;\lambda)=0$ if and only if $\Delta_1 \in [-1,1]$ or $\Delta_2 \in [-1,1]$. This observation along with Proposition \ref{detLem}
		provide the desired representation \eqref{abscontspectrum}. According to the Thomas’ analytic continuation argument, eigenvalues correspond to the constant branches of the dispersion relation \cite{KP07,RS78,T73}. Since the dispersion surfaces
		\begin{equation}
			\big\{(\Theta,\lambda) \in \mathbb{R}^3 ~|~ \Delta_{k}(\lambda) = \pm |s_0(\Theta)|/3 \text{ for some } k=1,2  \big\}
		\end{equation}
		have no constant branches outside $\Sigma^D$, we get $\sigma_{\text{pp}}(\cH) \subseteq \Sigma^D$ and hence 
		\begin{equation}
			\sigma_{\text{ac}}(\cH) = \{ \lambda \in \mathbb{R} ~|~ \Delta_k(\lambda) \in [-1,1]  \text{ for some } k=1,2\}.
		\end{equation}
		Note that \eqref{abscontspectrum} also represents $\sigma(\cH^{\text{per}}) = \sigma_{\text{ac}}(\cH^{\text{per}})$ by item (ii) in Theorem \ref{summaryRefResults}. So, the absolutely continuous spectrum $\sigma_{\text{ac}}(\cH)$ has band-gap structure and coincides as a set with the spectrum $\sigma(\cH^{\text{per}})$ of operator $\cH^{\text{per}}$ with potential $q_0$ periodically extended from $[0,1]$. Finally for item (iii), we observed that $\sigma_{\text{pp}}(\cH) \subseteq \Sigma^D$ and in Lemma \ref{LemmasigmaD} we showed that $\Sigma^D \subseteq \sigma_{\text{pp}}(\cH)$. Then Lemma \ref{Dspecsubsetpapspec} implies that $\sigma_{\text{pp}}(\cH) \subset \Sigma^{\text{p}}\cup\Sigma^{\text{ap}}$, where $\Sigma^{\text{p}}$ and $\Sigma^{\text{ap}}$ denote the periodic and anti-periodic spectra of \eqref{eigenvalueEquation}, i.e. with the boundary conditions \eqref{perbc} and \eqref{aperbc} respectively. However, from item (iii) in Theorem \ref{summaryRefResults} there exists $n_0\in \bN$ such that for all $n \geq n_0$ the edges of the n-th spectral band are the n-th periodic and anti-periodic eigenvalues. This concludes the proof.
	\end{proof}
	Next theorem proves existence of Dirac points, also called diabilical points, in the dispersion relation of $\cH$, where its different sheets touch to form a conical singularity.
	\begin{thm}
		\label{DiracPointsThm}
		\emph{\textbf{(Dirac points)}}
		The set of Dirac points of $\cH$ in the (first) Brillouin zone is
		\begin{equation*}
		\begin{split}
		\big\{(\Theta,\lambda) \in \mathbb{R}^3 ~|~ \Theta = \pm(2\pi/3,- 2\pi/3) \text{, } &T_2(\lambda-\eps,\lambda+\eps) \subset [0,\infty) \text{ and } \\
		&\Delta_k(\lambda)=0 \text{ for some } \eps > 0, k\in\{1,2\}\big\}.
		\end{split}
		\end{equation*}
	\end{thm}
	\begin{proof}[\normalfont \textbf{Proof of Theorem~\ref{DiracPointsThm}}] 
	    If $T_2(\lambda-\eps,\lambda+\eps) \not \subset [0,\infty)$, then $\lambda$ can not belong to the interior of a spectral band. If it is an edge of a band, it can not be a Dirac point, since it may provide a one-sided conical singularity. Observe that function $|s_0(\Theta)|$ on $[-\pi,\pi]^2$ has vanishing conical singularities at points $\pm (2\pi/3,-2\pi/3)$. From item (vii) in Theorem \ref{summaryRefResults} we know for $k=1,2$ and $\lambda$ so that $\Delta_k(\lambda) \in [-1,1]$, then $\Delta_k$ is analytic and has non-zero derivative in the neighborhood of $\lambda$ restricted to the interior of the corresponding band. Therefore $\Delta_k$ is monotonic in any spectral band around any $\lambda$ satisfying $\Delta_k(\lambda)=0$, so using the dispersion relation of $\cH$ we get the set of Dirac points.  
	\end{proof}
	\begin{remark}
		\label{DiracPics}
		One can classify the Dirac points $(\pm\Theta^{*},\lambda^*)$ with $\Theta^{*}:=(2\pi/3,- 2\pi/3)$ of the dispersion relation as follows:
		\begin{itemize}[leftmargin=*]
			\item If $\lambda^*$ is not a resonance point (i.e. $T_2(\lambda^*) \neq 0$) and $\Delta_k(\lambda^*) = 0$ for some $k \in \{1,2\}$, then the dispersion relation around each of the singularities $(\pm\Theta^{*},\lambda^*)$ consists of two cones located in opposite directions in $\lambda^*$-axis with the common vertex singularity $(\pm\Theta^{*},\lambda^*)$. See Figure \ref{fig:DiracTypes} (left). This is the case for large $\lambda^*$, i.e. high energy level scheme.  
			\item If $\Delta_1(\lambda^*) = \Delta_2(\lambda^*) = 0$ and there exists $\delta >0$ so that $|T_2(\lambda)| < 1 $ for all $\lambda \in [\lambda^*-\delta,\lambda^*+\delta]$, and $T_1(\lambda^*-\lambda) \not = T_1(\lambda^*+\lambda)$ for $\lambda \in (0,\delta)$, then  dispersion relation around each of the singularities $(\pm\Theta^{*},\lambda^*)$ consists of four cones, two of them located in opposite directions than the other two on $\lambda$-axis with the common vertex singularity at $(\pm\Theta^{*},\lambda^*)$. See Figure \ref{fig:DiracTypes} (right). Note that if $T_1(\lambda^*-\lambda) = T_1(\lambda^*+\lambda)$ for $\lambda \in (0,\delta)$, then the pairs of cones which are in the same directions coincide, so we get the first item above. 
		\end{itemize}
		\begin{figure}[ht]
			\centering
			\includegraphics[width=0.75\textwidth]{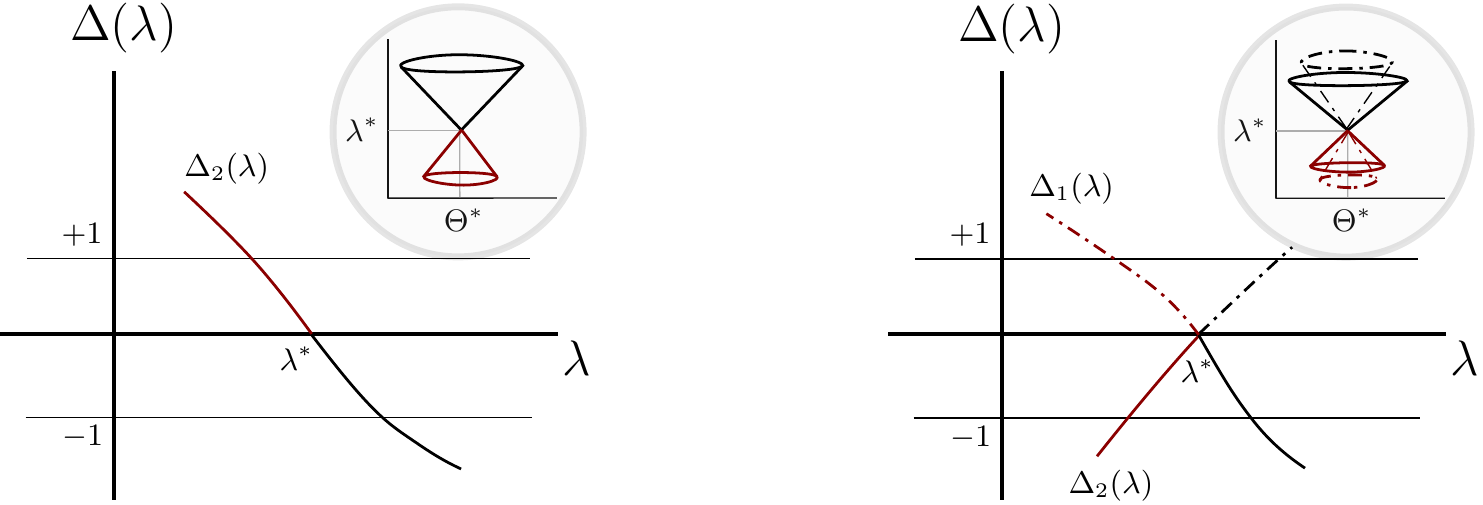}
			\caption{Behaviour of functions $\Delta_1$ and $\Delta_2$ near Dirac point $\lambda^*$. The circular windows schematically show the dispersion relation in a neighborhood of $(\pm\Theta^{*},\lambda^*)$, see Remark \ref{DiracPics} for details.}
			\label{fig:DiracTypes}
		\end{figure}
	\end{remark}
	
	Next result of this section is about irreducibility of Fermi surfaces corresponding to graphene Hamiltonian $\cH$ at high-energy levels. Depending on potential, reducibility of this surface may happen for uncountably many (low) energies. This is unlike special cases e.g. 2D and 3D discrete Laplacian plus a periodic potential, continuous Laplacian with special type of potential, and more general graph operators, where the underlying graph is planar with two vertices per period, in which irreducibility happens for all but finitely many energies \cite{LS20,FWS21}. Reducibility is required for the existence of embedded eigenvalues engendered by local defect, except for the anomalous situations when an eigenvalue has compact support \cite{KV06}. In summary, Fermi surface of a 2-periodic operator at an energy $\lambda$ is the set of wavevectors $(\theta_1,\theta_2)$ admissible by the operator at that energy. For periodic graph Hamiltonian, the dispersion function is a Laurent polynomial in the Floquet variables $(z_1,z_2) = (e^{i\theta_1}, e^{i\theta_2})$. When the dispersion function can be factored, for each fixed energy, as a product of two or more Laurant polynomials in $(\theta_1,\theta_2)$, each irreducible component contributes a sequence of special bands and gaps. We refer reader to the work \cite{FWS21} and references there for detailed discussions. By referring to Theorem \ref{detProp2}, the dispersion relation (Fermi surface) of $\cH$ is equivalent to the fact that $|s_0(\Theta)|^2/9$ is an eigenvalue of $\mathbb{G}_0^2(\lambda)$, i.e. it is a root of polynomial
	\begin{equation*}
		z^2 - \tr(\mathbb{G}_0^2(\lambda))z + \det(\mathbb{G}_0^2(\lambda)).
	\end{equation*}
	The roots of this quadratic polynomial have forms
	\begin{equation*}
		\frac{|s_0(\Theta)|^2}{9} = \frac{\tr(\mathbb{G}_0^2(\lambda))}{2} \pm \frac{1}{2} \big(\tr^2(\mathbb{G}_0^2(\lambda)) - 4\det(\mathbb{G}_0^2(\lambda))\big)^{1/2}.
	\end{equation*}
	Now observe that 
	\begin{align*}
		\frac{|s_0(\Theta)|^2}{9} &= 
		\frac{\tr^2(\mathbb{G}_0(\lambda))}{2} - \det(\mathbb{G}_0(\lambda)) \pm \frac{1}{2}\big(\tr(\mathbb{G}_0(\lambda))\big(\tr^2(\mathbb{G}_0(\lambda)) - 4\det(\mathbb{G}_0(\lambda))\big)^{1/2}\big)\\
		&=- \det(\mathbb{G}_0(\lambda)) + \frac{1}{2} \tr(\mathbb{G}_0(\lambda))
		\big(\tr(\mathbb{G}_0(\lambda)) \pm \big(\tr^2(\mathbb{G}_0(\lambda)) - 4\det(\mathbb{G}_0(\lambda))\big)^{1/2} \big).
	\end{align*}
	Application of $T_1(\lambda)$ and $T_2(\lambda)$ from \eqref{eq:T1T2Def} in $\Delta_k$,  implies that 
	\begin{align*}
		\frac{|s_0(\Theta)|^2}{9} 
		= T_1^2(\lambda) + T_2(\lambda) \pm 2T_1(\lambda)T_2^{1/2}(\lambda)
		= \Delta_{1,2}^2(\lambda).
	\end{align*}
	So we proved the following result on reducibility of the Fermi surface of $\cH$. 
	\begin{thm}\label{ReducibleFermisurface}\emph{\textbf{(Fermi surfaces)}} 
		The relation \eqref{eq:DispRelation} has representation 
		\begin{equation*}
			\big(P(z_1,z_2) P(z_1^{-1},z_2^{-1}) - 9\Delta_1^2(\lambda)\big)\big(P(z_1,z_2) P(z_1^{-1},z_2^{-1}) - 9\Delta_2^2(\lambda)\big) = 0,
		\end{equation*}
		where $P(\omega_1,\omega_2) := 1+\omega_1+\omega_2$ and $z_1 = e^{i\theta_1}$ and $z_2 = e^{i\theta_2}$. Moreover, letting $\cS_1 := \{ \lambda \in \mathbb{R} ~|~ \Delta_1(\lambda) \in [-1,1] \}$ and $\cS_2 := \{ \lambda \in \mathbb{R} ~|~ \Delta_2(\lambda) \in [-1,1] \}$, then Fermi surface with the energy level $\lambda \not \in \Sigma^D$ is 
		\begin{itemize}
			\item reducible if $\lambda \in (\cS_1\cap \cS_2)$,
			\item irreducible if $\lambda \in (\cS_1\setminus \cS_2)\cup(\cS_2\setminus \cS_1)$,
			\item absent if $\lambda \in \bR\setminus (\cS_1\cup \cS_2)$.
		\end{itemize}
	\end{thm}
	\textbf{A Remark on Choices of Brillouin Zone:}	
	There exists some room on the choice of fundamental domain $W$ for hexagonal lattices. For the selected one in Figure \ref{fig:fundDomain}, space and quasimomentum (conjugate) basis with respect to global coordinate system are of the form  
		\begin{equation}
		\vec b_1 = \frac{1}{2}
		\begin{pmatrix}
			3\\
			\sqrt{3}
		\end{pmatrix}, 
		\qquad
		\vec b_2 = 
		\begin{pmatrix}
			0\\
			\sqrt{3}
		\end{pmatrix},	
		\qquad
		\vec b_1^* = \frac{2}{3}
		\begin{pmatrix}
			1\\
			0
		\end{pmatrix}, 
		\qquad
		\vec b_2^* =  \frac{1}{3}
		\begin{pmatrix}
			-1\\
			\sqrt{3}
		\end{pmatrix}.
	\end{equation}
	The dual basis then satisfies 
	\begin{equation}
		\label{eq:orthon}
		\vec b_n^* \cdot \vec b_m = \delta_{nm}
	\end{equation}
	and vectors $2\pi \vec b_1^* $ and $2\pi\vec b_2^* $ span hexagonal lattice as well denoted by $\Gamma^*$. Now the orthonormality condition \eqref{eq:orthon} implies
	\begin{equation}
		n_1 \theta_1 + n_2 \theta_2 = \big(\theta_1 \vec b_1^* + \theta_2 \vec b_2^*\big) \cdot \big(n_1 \vec b_1 + n_2 \vec b_2\big).
	\end{equation}
	The two choices of Brillouin zone using coordinates $\Theta = (\theta_1, \theta_2)$ with respect to dual basis vectors $\vec b_1^*, \vec b_2^*$ are shown in Figure \ref{fig:Brillouin}. In the literature it is more common to represent these Brillouin zones in corresponding Cartesian coordinates $\vec \kappa = (k_1, k_2)^T$ given by $\vec \kappa = B^*\Theta$, where $B^*$ is the transformation matrix with columns formed by the dual basis vectors, i.e. 
	\begin{equation}
		B^* = 
		\begin{pmatrix}
			\vec b_1^* ~
			\vec b_2^* 
		\end{pmatrix}
		= \frac{1}{3}
		\begin{pmatrix}
			2 & -1\\
			0 & \sqrt{3}
		\end{pmatrix}.
	\end{equation}
	\begin{figure}[ht]
		\centering
		\includegraphics[width=0.75\textwidth]{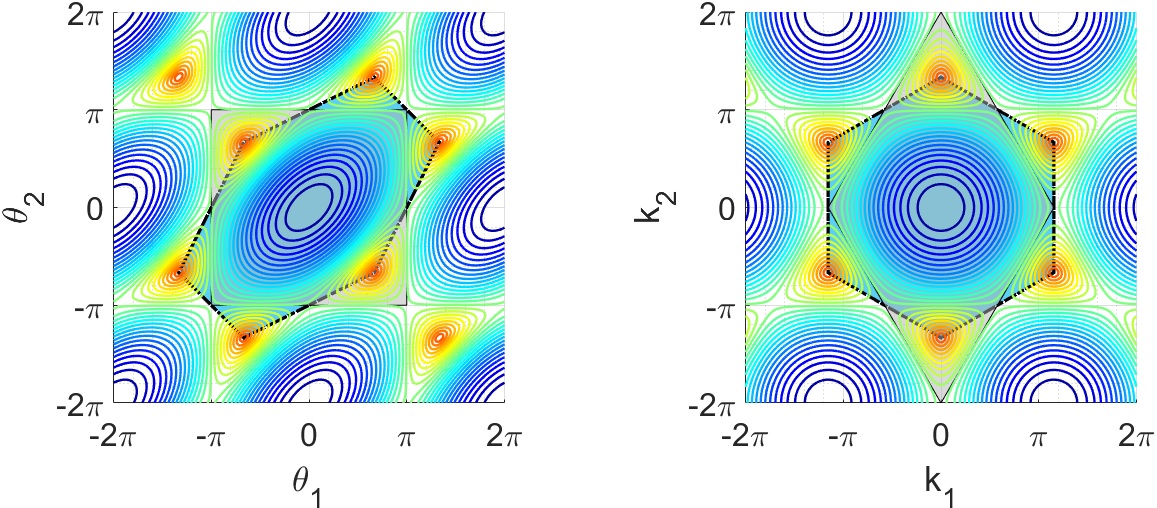}
		\caption{Choices of Brillouin zone, contour plot of second sheet of the dispersion surface in Left: coordinates $\theta_1, \theta_2$ (drawn as if they were Cartesian) and Right: coordinates $k_1$, $k_2$ (which are Cartesian). }
		\label{fig:Brillouin}
	\end{figure}
	
	As it is shown in Figure \ref{fig:Brillouin}(right), the resulting Brillouin zones will be symmetric in the new coordinates system $\vec \kappa$.  One arrives at the first picture by using $\theta_1$ and $\theta_2$ as parameters for the dispersion relation ranging within $[-\pi, \pi]^2$  and then plots the result using $k_1$ and $k_2$ as Cartesian coordinates. Although these two representations are equal, for symmetry discussion it may be more preferable to work with  $\vec \kappa$ coordinate system, while for our case we followed the Brillouin zone in $\Theta$ coordinates due to simpler presentation of the vertex conditions, see  \eqref{eq:cond1At1}-\eqref{eq:cond4At1}. Interested reader is encouraged to look at the work \cite{BC18} for detailed discussions. 
	
	\section{Perturbed Hamiltonian}
	\label{sec:APH}
	In this section we will apply tools from perturbation theory to characterize dispersion relation for the case in which edges meet at generally different angels, see Figure \ref{fig:anglePrturb} for schematic fundamental domains. Restricted to fundamental domain $W$, this is equivalent to find $(\lambda, \Theta) \in \bR \times [-\pi ,\pi]^2$ so that $\det(\mathbb{M}_\eps(\lambda)) = 0$ as stated in Proposition \ref{detProp}.
	\begin{figure}[ht]
		\centering
		\includegraphics[width=0.8\textwidth]{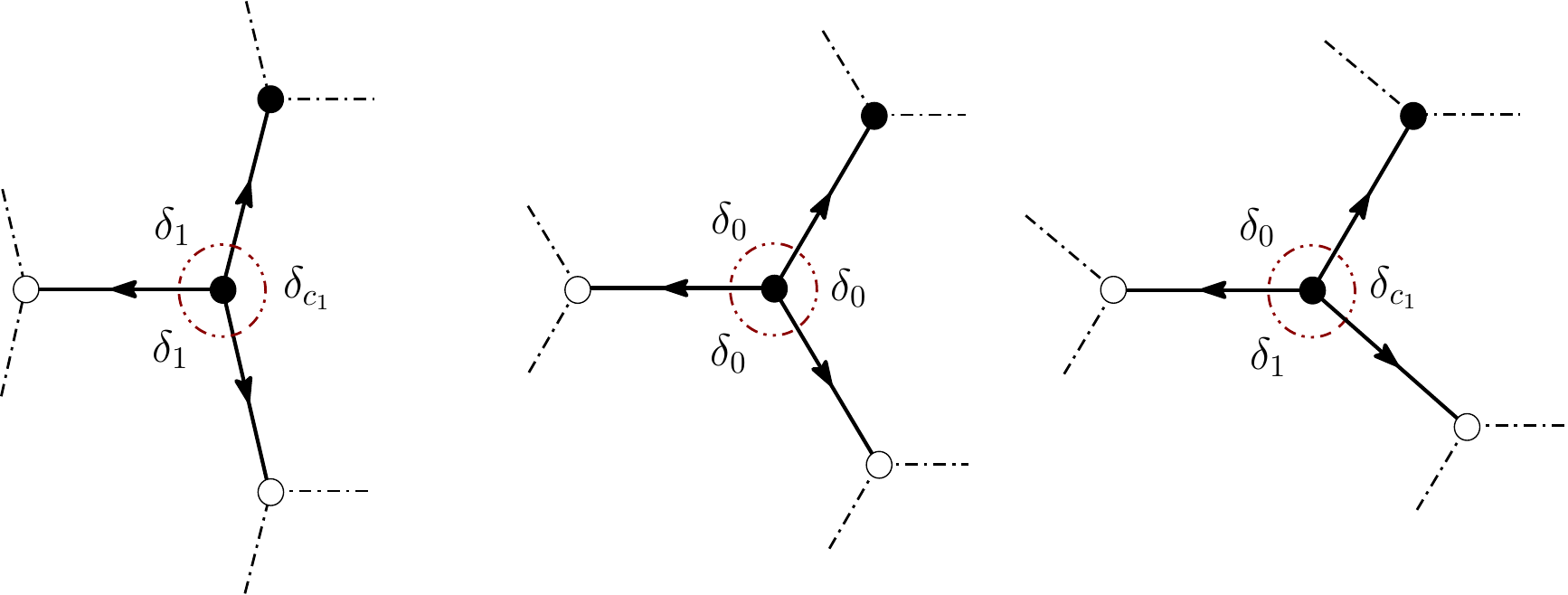}
		\caption{Fundamental domains for angle perturbed hexagonal lattices, see \eqref{eq:anglePerturb}. The middle picture shows the graphene lattice in which edges are met with equal angles at each vertex.}
		\label{fig:anglePrturb}
	\end{figure}
	
	First observe that for angle $\delta_{c}\up{\eps}$, see \eqref{eq:anglePerturb}, an expansion of sine function has the form 
	\begin{equation}
		\sin(\delta_{c}\up{\eps}) = \sin(\delta_0) + \eps c \cos(\delta_0) + \cO(\eps^2)
	\end{equation}
	as $\eps$ goes to zero. A similar result holds as
	\begin{equation}
		\sin^2(\delta_{c}\up{\eps}) = \sin^2(\delta_0) + 2\eps c \cos^2(\delta_0) + \cO(\eps^2).
	\end{equation}
	Let us introduce 
	\begin{equation}
		\label{eq:s1}
		s_1(\Theta) := \cot(\delta_0) (1 + c_1 e^{-i \theta_1} + c_2 e^{-i \theta_2}).
	\end{equation}
	Then up to order $\cO(\eps^2)$ accuracy, $\mathbb{M}_\eps(\lambda)$ has expansion of the form 
	\begin{equation*}
		\mathbb{M}_{\eps} := \mathbb{M}_0 + \eps \mathbb{M}_1 + O(\eps^2),
	\end{equation*}
	in which the two matrices have block structures as
	\begin{equation*}
		\mathbb{M}_0 := \begin{pmatrix}
			s_{0}(0)\Phi_0(0) & - s_{0}(\Theta)\Phi_0(1)\\
			- \overline{s_{0}(\Theta)}\Phi_0(1) & s_{0}(0)\Phi_0(0)
		\end{pmatrix}, \qquad \mathbb{M}_1:= \begin{pmatrix}
			s_1(0)\Phi_{1}(0) & - s_1(\Theta)\Phi_{1}(1)\\
			- \overline{s_1(\Theta)}\Phi_{1}(1) & s_1(0)\Phi_{1}(0)
		\end{pmatrix}
	\end{equation*}
	with $2 \times 2$ blocks 
	\begin{align*}
		\Phi_0(0) &:= \begin{pmatrix}
			\phi_1'(0) & \phi_2'(0)\\
			\phi_1'''(0) & \phi_2'''(0) 
		\end{pmatrix}, \qquad \Phi_0(1) := \begin{pmatrix}
			\phi_1'(1) & \phi_2'(1)\\
			\phi_1'''(1) & \phi_2'''(1) 
		\end{pmatrix},\\
		\Phi_1(0) &:= \begin{pmatrix}
			\phi_1'(0) & 2\phi_2'(0)\\
			0 & \phi_2'''(0) 
		\end{pmatrix}, \qquad \Phi_1(1) := \begin{pmatrix}
			\phi_1'(1) & 2\phi_2'(1)\\
			0 & \phi_2'''(1) 
		\end{pmatrix}.
	\end{align*}
	Applying the fact that $\Phi_0(1)$ is non-singular, see Lemma \ref{A1NonSingular}, we introduce
	\begin{equation}
		\Lambda_0(0) := \Phi_0^{-1}(1) \Phi_0(0), \qquad 	\Lambda_1(1) := \Phi_0^{-1}(1) \Phi_1(1).
	\end{equation}
	Then up to $\cO(\eps^2)$ error, the perturbed matrix $\mathbb{M}_\eps$ can be explicitly written as
	\begin{equation}
		\label{eq:Meps}
		\mathbb{M}_\eps(\lambda) = 
		\begin{pmatrix}
			3 \Lambda_0(0) & -s_0(\Theta) \\
			-\overline{s_0(\Theta)}& 3 \Lambda_0(0) 
		\end{pmatrix}
		+ \eps 
		\begin{pmatrix}
			0 & - s_1(\Theta) \Lambda_1(1) \\
			- \overline{s_1(\Theta)} \Lambda_1(1) & 0
		\end{pmatrix}.
	\end{equation}
	As stated in Theorem \ref{detProp2}, equality $\mathbb{G}_0(\lambda) = \Lambda_0(0)$ holds with components of $\mathbb{G}_0(\lambda)$ in terms of the fundamental solutions
	\begin{equation}
		\mathbb{G}_0(\lambda) =
		\begin{pmatrix}
			g_1(1) & g_3(1)\\
			g_1''(1) & g_3''(1) 
		\end{pmatrix}.
	\end{equation}
	Denoting by $\widetilde{\cD}(f,g) := f(1)g''(1) - g(1)f''(1)$, we represent functions $\phi_1$ and $\phi_2$ in terms of the fundamental solutions. 
	\begin{lem}\label{FundSoltoPhi}
		Functions $\phi_1$ and $\phi_2$ have representations
		\begin{align*}
			\phi_1(x) &= g_1(x) + \widetilde{\cD}^{-1}(g_2,g_4) \big(\widetilde{\cD}(g_4,g_1) g_2(x) + \widetilde{\cD}(g_1,g_2)g_4(x)\big),\\
			\phi_2(x) &= g_3(x) + \widetilde{\cD}^{-1}(g_2,g_4)\big(\widetilde{\cD}(g_4,g_3) g_2(x) + \widetilde{\cD}(g_3,g_2)g_4(x)\big).
		\end{align*}
	\end{lem}
	Application of Lemma \ref{FundSoltoPhi} along with characterization $\phi_2'''(1) = \phi_1'(1)$ yields representation of $\Lambda_1(1)$ in \eqref{eq:Meps} in terms of the fundamental solutions. From now on we call this representation $\mathbbm{G}_1(\lambda)$ matrix. One way to calculate determinant of $\mathbbm{M}_\eps(\lambda)$ is to apply results on analysis of perturbed matrices, e.g. see \cite{K95} and references there. However, we calculate this quantity directly up to $\cO(\eps^2)$ order, which under heavy simplification of the terms turns to be 
	\begin{equation}
		\label{eq:DetMeps}
		\det(\mathbbm{M}_\eps) =d_0 + \eps d_1+ \cO(\eps^2).
	\end{equation}
	The $d_0$ is equal to determinant of $\mathbb{M}_0$ matrix 
	\begin{equation}
		\begin{split}
			d_0 = \det(\mathbbm{M}_0) = \frac{|s_0(\Theta)|^4}{81} - \frac{|s_0(\Theta)|^2}{9} \tr(\mathbbm{G}_0^2) + \det(\mathbbm{G}_0^2) 
		\end{split}.
	\end{equation}
	Moreover, the $\eps$-contribution term is
	\begin{equation}
		d_1 = -4\frac{|s_0(\Theta)|^2}{9}\Re\big(s_0(\Theta) \overline{ s_1(\Theta)}\big) G(\lambda)
	\end{equation}
	with the purely $\lambda$-dependent function 
	\begin{equation}
		G(\lambda) = -\frac{1}{2} \big\{(1-(\mathbbm{G}_0^2)_{22})(\mathbbm{G}_1)_{11} + (1-(\mathbbm{G}_0^2)_{11}) (\mathbbm{G}_1)_{22} + 
		(\mathbbm{G}_0^2)_{21}(\mathbbm{G}_1)_{12} +
		(\mathbbm{G}_0^2)_{12}(\mathbbm{G}_1)_{21} \big\}.
	\end{equation}
	Therby up to $\eps^2$ accuracy, zeros of perturbed determinant \eqref{eq:DetMeps} is equivalent to the fact that $|s_0(\Theta)|^2/9$ is a root of polynomial 
	\begin{equation}
		\label{eq:poly1}
		\cP(z) = z^4 -\big(\tr(\mathbbm{G}_0^2)+4\eps \Re(s_0(\Theta) \overline{ s_1(\Theta)}) \xi(\lambda)\big)z^2 + \det(\mathbbm{G}_0^2). 
	\end{equation}
	Notice that a fourth-order polynomial of form $z^4 -az^2 + b$ can be factorized as
	\begin{equation}
		z^4 -az^2 + b = (z^2+\tilde a z+ \tilde b)(z^2-\tilde a z+ \tilde b)
	\end{equation}
	in which $\tilde a= (a+2b^{1/2})^{1/2}$ and $\tilde b = b^{1/2}$. This realization along with the form \eqref{eq:poly1} implies that $\pm|s_0(\Theta)|/3$ are roots of $\cP(z) = \cP_1(z)\cP_2(z)$, where
	\begin{equation}
		\begin{split}
			\cP_{1,2}(z) = z^2 \pm\big(\tr(\mathbbm{G}_0^2)+2\text{det}^{1/2}(\mathbbm{G}_0^2) 
			&+4\eps \Re(s_0(\Theta) \overline{s_1(\Theta)})G(\lambda) \big)^{1/2}z \\+ 
			\big(\tr(\mathbbm{G}_0^2)&+4\eps \Re(s_0(\Theta) \overline{ s_1(\Theta)})G(\lambda)\big)^{1/2} 
		\end{split}.
	\end{equation}
	Without loss of generality, let's assume that $|s_0(\Theta)|/3$ is a root of $\cP_2$, i.e. 
	\begin{equation}
		\begin{split}
			\frac{2}{3} |s_0(\Theta)| =  &\big(\tr(\mathbbm{G}_0^2)+2\text{det}^{1/2}(\mathbbm{G}_0^2) 
			+4\eps \Re(s_0(\Theta) \overline{s_1(\Theta)})G(\lambda) \big)^{1/2} \pm  \\
			&\big(\tr(\mathbbm{G}_0^2)-2\text{det}^{1/2}(\mathbbm{G}_0^2) 
			+4\eps \Re(s_0(\Theta) \overline{ s_1(\Theta)})G(\lambda) \big)^{1/2}.
		\end{split}
	\end{equation}
	Now applying the fact that 
	\begin{equation}
		\tr(\mathbbm{G}_0^2) = \tr^2(\mathbbm{G}_0) - 2\text{det}(\mathbbm{G}_0)
	\end{equation}
	along with equality $\text{det}^{1/2}(\mathbbm{G}_0^2) = \text{det}(\mathbbm{G}_0) $ implies that 
	\begin{equation}
		\begin{split}
			\frac{|s_0(\Theta)| }{3} =  &\big( \frac{1}{4}\tr^2(\mathbbm{G}_0) 
			+\eps\Re(s_0(\Theta) \overline{s_1(\Theta)})G(\lambda) \big)^{1/2} \pm  \\
			&\big(\frac{1}{4} \tr^2(\mathbbm{G}_0)-\text{det}(\mathbbm{G}_0) 
			+\eps \Re(s_0(\Theta) \overline{s_1(\Theta)})G(\lambda) \big)^{1/2}.
		\end{split}
	\end{equation}
	However using the definitions of $T_1$ and $T_2$ in \eqref{eq:T1T2Def}, we  introduce $\eps$-extension of these functions as 
	\begin{equation}\label{eq:T1T2eps1}
	    T_1\up{\eps} := \big(T_1^2(\lambda) +\eps\Re(s_0(\Theta) \overline{ s_1(\Theta)}) G(\lambda)\big)^{1/2}
	\end{equation}
	and
	\begin{equation}
	   	T_2\up{\eps} := T_2(\lambda)  +\eps \Re(s_0(\Theta) \overline{ s_1(\Theta)}) G(\lambda). 
	\end{equation}
	Finding the roots of quadratic polynomials $\cP_{1,2}$ is then reduces to condition $|s_0(\Theta)|/3$ satisfying 
	\begin{equation}
		\label{eq:T1T2eps2}
		\pm \frac{|s_0(\Theta)|}{3} = T_1\up{\eps}+ (T_2\up{\eps})^{1/2} \quad \text{or} \quad 
		\pm \frac{|s_0(\Theta)|}{3} = T_1\up{\eps} - (T_2\up{\eps})^{1/2}.
	\end{equation}
	Thus we proved an $\eps$-extended dispersion relation for perturbed Hamiltonian as stated below. 
	\begin{thm}\label{anglePertDispThm}\emph{\textbf{(perturbed dispersion)}} 
		The dispersion relation for perturbed graphene Hamiltonian up to accuracy $\cO(\eps^2)$ satisfies
		\begin{equation}
			\label{eq:D1D2eps}
			\Big(\Delta_{1}\up{\eps}(\lambda,\Theta) \pm \frac{|s_0(\Theta)|}{3}\Big)\Big( \Delta_{2}\up{\eps}(\lambda,\Theta) \pm \frac{|s_0(\Theta)|}{3} \Big)= 0,
		\end{equation}
		where $\Delta_{1,2}\up{\eps} := T_1\up{\eps} \pm (T_2\up{\eps})^{1/2}$. 
	\end{thm}
	We stress out here that for the case $\eps = 0$, results above are consistent with the ones stated for graphene Hamiltonian. One of the interesting futures of Theorem \ref{anglePertDispThm} is to answer whether singular Dirac points will be preserved under $\eps$-perturbed geometry. To answer this we first characterize the behaviour of $\Theta$-dependent function $ \Re(s_0(\Theta) \overline{ s_1(\Theta)})$ in perturbed part. 
	\begin{lem}
		\label{ReS0S1}
		Function $ \Re(s_0(\Theta) \overline{s_1(\Theta)})$ is $2\pi \bZ^2$ periodic, its magnitude is bounded by $2(1+|c_1|)$ and zeros are at $(0,0)$ and $\pm(2\pi/3, -2\pi/3)$. 
	\end{lem}
	\begin{proof}[\normalfont \textbf{Proof of Lemma~\ref{ReS0S1}}] 
		Recalling the definitions of $s_0(\Theta)$ and $s_1(\Theta)$ from \eqref{eq:s0} and \eqref{eq:s1} respectively
		\begin{equation}
			s_0(\theta) \overline{ s_1(\Theta)} = -\cot(\delta_0)(1+e^{-i\theta_1} + e^{-i\theta_2}) (1+c_1e^{i\theta_1} +c_2e^{i\theta_2}). 
		\end{equation}
		By representation of exponential terms using Euler-formula
		\begin{equation}
			\Re(s_0(\Theta) \overline{ s_1(\Theta)}) = -\cot(\delta_0) \big((1+c_1) \cos(\theta_1) + (1+c_2) \cos(\theta_2)  + (c_1+c_2) \cos(\theta_2- \theta_1) \big),
		\end{equation}
	which after further simplification and application of identity $1+c_1+c_2 = 0$ will reduce to
		\begin{equation}
			\label{eq:s0Bars1}
			\Re(s_0(\Theta)\overline{ s_1(\Theta)}) = -\cot(\delta_0) \big( \cos(\theta_2- \theta_1) + c_1 \cos(\theta_2) + c_2 \cos(\theta_1) \big).
		\end{equation}
		Applying the fact that $\cos(\delta_0) = \cos(2\delta_0)$, then $(0,0)$ and $\pm(2\pi/3, -2\pi/3)$ are zeros of the functions. Finally, setting $c_2 = -1-c_1$ above we get
		\begin{equation}
			|\Re(s_0(\Theta) \overline{ s_1(\Theta)}) | \leq | \cos(\theta_2- \theta_1) -\cos(\theta_1) + c_1\big(\cos(\theta_2)-\cos(\theta_1)\big)|  \leq 2(1+|c_1|)
		\end{equation}
		as desired. 
	\end{proof}
	Figure \ref{Res0Bars1}(right) shows the behaviour of function $\Re(s_0(\Theta)\overline{s_1(\Theta)})$ for a fixed value of $c_1$ parameter.  
	\begin{figure}[ht]
		\centering
		\includegraphics[width=0.85\textwidth]{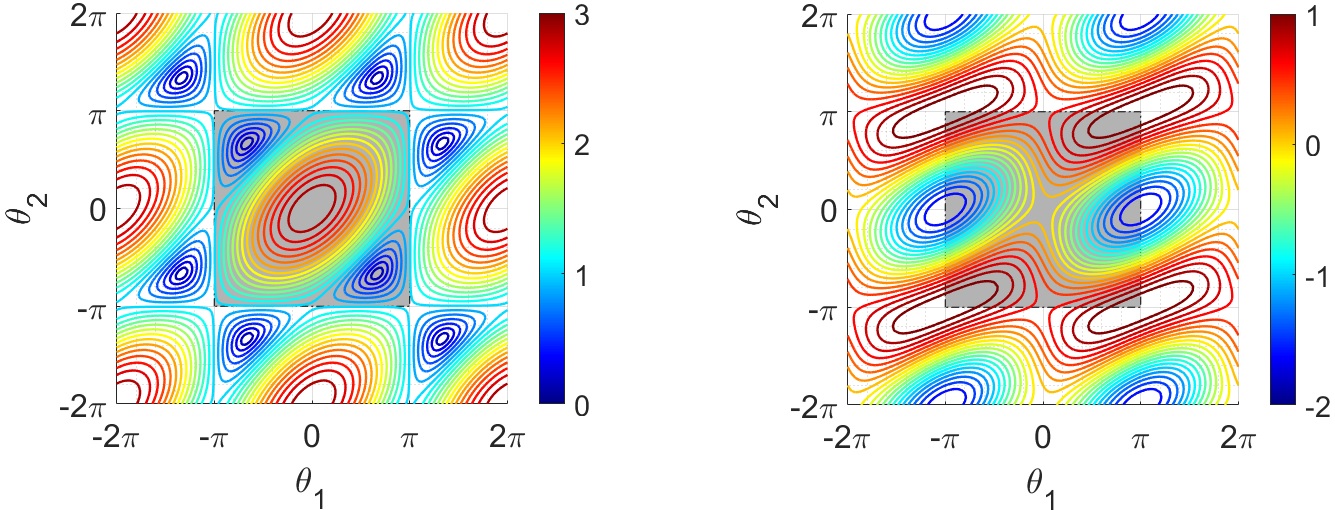}
		\caption{Plots of $|s_0(\Theta)|$ and $\Re(s_0(\Theta) \overline{ s_1(\Theta)})$. Highlighted rectangle shows the first Brillouin zone.}
		\label{Res0Bars1}
	\end{figure}
	\begin{cor}\label{DiracPerturb}\emph{\textbf{(Dirac points)}} 
		Singular Dirac points of graphene remain under angle-perturbed Hamiltonian. 
	\end{cor}
	\begin{proof}[\normalfont \textbf{Proof of Corollary~\ref{DiracPerturb}}] 
		Let $(\Theta^*, \lambda^*)$ with $\Theta^* := \pm(2\pi/3, -2\pi/3)$ be a Dirac point for graphene Hamiltonian. Then, result from Lemma \ref{ReS0S1} implies that function $ \Re(s_0(\Theta) \overline{s_1(\Theta)})$ vanishes on quasimomentum $\Theta^*$. Thereby,  there is no spectral gap at energy $\lambda^*$ for the perturbed Hamiltonian as well. Regarding singularity at this point, let define $\eps$-dependent function
		\begin{equation}
			D_\eps(\lambda,\Theta) := \pm \frac{|s_0(\Theta)|}{3} - T_1\up{\eps} - (T_2\up{\eps})^{1/2}
		\end{equation}
		and similarly for $\Delta_2\up{\eps}$. Applying continuity property of function $D_\eps(\lambda,\Theta)$ with respect to $\Theta$, then there exist $\eps$-dependent neighborhood $\cN_{\lambda,\Theta}\up{\eps} := \cN_{\lambda^*}(\lambda) \times \cN_{\Theta^*}\up{\eps}(\Theta)$ containing $(\lambda^*, \Theta^*)$ so that $D_\eps(\lambda,\Theta)$ is well defined for all $(\lambda,\Theta) \in \cN_{\lambda,\Theta}\up{\eps} $. For $\lambda \in \cN_{\lambda,\Theta}\up{\eps} \setminus\{\lambda^*\}$ and the case $T_2(\lambda) > 0$, then application of inverse function theorem implies that solution set for $D_\eps(\lambda,\Theta) = 0$ is a simple closed loop (distorted ellipse) in quasimomentum $\cN_{\Theta^*}\up{\eps}(\Theta)$. Moreover, observe that singularity of function $D_\eps(\lambda,\Theta) $ only occurs at $\Theta^*$ due to $|s_0(\Theta)|$. For the case $T_2(\lambda) = 0$, function $D_\eps(\lambda,\Theta)$ is only well-defined for $\cN_{\Theta^*}\up{\eps}(\Theta) \cap \{\Theta : \Re(s_0(\Theta) \overline{ s_1(\Theta)})G(\lambda) \geq 0\} $. Similar discussion implies that solution set for $D_\eps(\lambda,\Theta) = 0$ is a simple connected curve (not closed) in quasimomentum $\cN_{\Theta^*}\up{\eps}(\Theta)$. In this case, dispersion relation is lost locally for $\Theta$ such that  $\Re(s_0(\Theta) \overline{ s_1(\Theta)})G(\lambda) < 0$. In all two cases, the gap remain closed at Dirac point, however only one-side differentiability exists for the latter case.
	\end{proof}
	\begin{remark}
		Here we stress out that for the case $T_2(\lambda) = 0$ explained in the proof of Corollary \ref{DiracPerturb}, concern is only about $\lambda \not = \lambda^*$ as for $\Theta^*$ the $\eps$-term vanishes. Moreover, Corollary \ref{DiracPerturb} guarantees that Dirac points appear at $\Theta^* = \pm(2\pi/3, -2\pi/3)$ quasimomenta, but with possible shift in the energy space. 
	\end{remark}
    Investigation on presence of a pure point spectrum has been an active research area. Changing the geometry of medium (e.g. working with 2D periodic graph instead of a real line), imposing perturbation through potential, and applying different Hamiltonian model are among few ways to guarantee presence of a pure point spectrum, e.g. see \cite{K03,KP07,H18,L21}. As stated in Theorem \ref{grapheneSpectrum} pure point spectrum for graphene is non-empty. This has been proved by explicit construction of even (or odd) eigenfunctions with support on single hexagon. However, existence of pure point spectrum will fail for perturbed Hamiltonian.  
	\begin{thm}\label{PerturbedgrapheneSpectrum}\emph{\textbf{(the spectral description)}}
		The spectrum of the perturbed  Hamiltonian is purely absolutely continuous.
	\end{thm}
	\begin{proof}[\normalfont \textbf{Proof of Theorem~\ref{PerturbedgrapheneSpectrum}}]
		The singular continuous spectrum is empty, since the Hamiltonian is a self-adjoint elliptic operator like the unperturbed case (see proof of Theorem \ref{grapheneSpectrum}). Next, let's show the absence of the pure point spectrum unlike the graphene case. Using the dispersion relation, we get $\sigma_{\text{pp}}(\cH) \subset \Sigma^{D}$ like we did in the unperturbed case. Now let us assume $\sigma_{\text{pp}}(\cH) \neq \emptyset$. Then the corresponding eigenfunction $u$ restricted to any edge should either be identically zero or solve $d^4u(x)/dx^4 + q(x)u(x) = \lambda u(x)$ with the boundary conditions $u(0)=u(1)=u''(0)=u''(1)=0$ on that edge. Therefore restriction of an eigenfunction to any edge on its support should be an eigenfunction of the operator $d^4/dx^4 + q(x)$ for the same eigenvalue $\lambda$ on $[0,1]$ interval with the boundary conditions $u(0)=u(1)=u''(0)=u''(1)=0$. Note that $u$ should also satisfy the vertex conditions. 

		If $u$ is compactly supported, then the vertex conditions on the vertices of the boundary of the support of $u$ imply $\eps = c_1\eps = c_2\eps$. Recall that $1 + c_1 + c_2 = 0$, so $\eps = 0$ is the only solution, which is the unperturbed case. For the non-compactly supported $u \in H^4(\Gamma)$, same discussion holds to show that vertex conditions can not be met at any vertex. Therefore the pure point spectrum is also empty. From the dispersion relation we get that the spectrum is non-empty, so we get the desired result that the spectrum is purely absolutely continuous. 
	\end{proof}
	\begin{remark}
		Applying the result in Lemma \ref{ReS0S1} and the proof of Corollary \ref{DiracPerturb} some arguments can be made to quantify shift of dispersion relation  \eqref{eq:D1D2eps} for perturbed Hamiltonian compared to graphene case at any $\lambda$. More precisely, for $T_2 > 0$, expansion of $T_1\up{\eps}$ and $T_2\up{\eps}$ in  \eqref{eq:T1T2eps2}  implies that 
		\begin{equation}
			\label{eq:qusiMomeps}
			\pm\frac{|s_0(\Theta)|}{3} = \Delta_1(\lambda) \Big\{1+ \eps \Re(s_0(\Theta) \overline{ s_1(\Theta)})G(\lambda)T_1^{-1}(\lambda)T_2^{-1/2}(\lambda) \Big\} + \cO(\eps^2)
		\end{equation}
		and similarly for $\Delta_2$ with sign changes. Now for fixed value of $\lambda$, the shift with respect to graphene, i.e. case $\eps = 0$, in quasimomentum can be found by solving \eqref{eq:qusiMomeps}. 
	\end{remark}
	Finally, in the following section we give a partial list of topics which may be interesting to the reader for future extension of current work. 
	\section{Outlook}
	\label{sec:Outlook}
	The viability of the frame model as a structure composed of one-dimensional segments needs to be verified mathematically, as a limit of a three-dimensional structure as the beam widths go to zero. There is a significant mathematical literature on this question for second-order operators (see, for example \cite{Z02,Gri_incol08,Post_book12}), with a variety of operators arising in the limit.  This variety will increase in the case of fourth-order equations, and may be expected to incorporate masses concentrating at joints and other cases of applied interest. Moreover, validity of Euler-Bernouli beam theory specially at high-energy level maybe a place to be questioned. Unlike this, richer Timoshenko model no longer assumes the cross-sections remain orthogonal to the deformed axis and therefore incorporates more degrees of freedom \cite{MPZ02,GR15,Mei19}. Of applied interest would be to extend the current results to the latter model. 
	
	In this work we focused on Euler-Bernouli beam theory and its restriction to scalar-valued lateral displacement $v(x)$. In the work \cite{BE21} it is shown that for planar graphs, more accurate way to present the operator is by including angular displacement field $\eta(x)$ as well. This then shifts our problem to a vector-valued operator and more complicated vertex conditions. We refer to recent work \cite{SBE21} for analysis in this line and potential future work for interesting three dimensional periodic graphs.
	
	An interesting problem is to employ two-scale analysis for understanding the homogenized behavior and spectra of Hamiltonian on periodic lattices with more complex fundamental domain, e.g. see \cite{K16,ZP16,GMO18,CLM20,KM21} and references there. Of similar interest is generalization of our result to multi-layer quantum graph models equipped with beam Hamiltonian. In the work \cite{FWS21}, it is shown that for {S}chr{\"o}dinger operator dispersion relation of wave vector and energy is a polynomial in the dispersion relation of the single layer. This leads to the reducibility of the algebraic Fermi surface, at any energy, into several components. For the beam Hamiltonian, it has been shown that in the special case of planar frames, the operator decomposes into a direct sum of two operators, one coupling out-of-plane displacement to angular displacement and the other coupling in-plane displacement with axial displacement \cite{BE21}. Understanding the interaction of these  decoupled systems on multi-layer graphs may be interesting from both theory and applied perspectives. 
	
	\bibliographystyle{abbrv}
	\bibliography{ref}

\end{document}